\documentclass[onecolumn,pra, superscriptaddress,longbibliography,
12pt,tightenlines,aps]{revtex4-2}

\pdfoutput=1

\usepackage{graphicx}
\usepackage{bm}
\usepackage{amsmath}
\usepackage{environ}
\usepackage{appendix}
\usepackage{physics}
\usepackage{xfrac}
\usepackage[shortlabels]{enumitem}
\usepackage[nopar]{lipsum}
\usepackage{optidef}
\usepackage{stmaryrd}
\usepackage{xcolor}
\usepackage{bm}

\usepackage{amsthm}
\usepackage{amssymb}
\usepackage{amsfonts}
\usepackage{slashed}
\usepackage{bbm}
\usepackage{latexsym,epsfig,bbm}
\usepackage[makeroom]{cancel}

\usepackage[colorlinks = true, linkcolor = blue, urlcolor=blue, citecolor = red, anchorcolor = green]{hyperref}

\newcommand{\cM}{\mathcal{M}}
\newcommand{\cN}{\mathcal{N}}
\newcommand{\cA}{\mathcal{A}}

\usepackage{color}
\definecolor{blue}{rgb}{0,0.2,1}

\definecolor{red}{rgb}{0.9,0,0}

\newtheorem{theorem}{Theorem}
\newtheorem{lemma}[theorem]{Lemma}
\newtheorem{definition}[theorem]{Definition}

\newtheorem{corollary}[theorem]{Corollary}
\newtheorem{remark}{Remark}
\newtheorem{proposition}[theorem]{Proposition}

\usepackage{mathtools}

\newcommand{\Sn}{\mathfrak{S}_n}
\NewDocumentCommand{\End}{o}{\mathrm{End}^{\Sn}}

\allowdisplaybreaks
\newcommand{\comment}[1]{}

\begin{document}

 \title{Query Complexity of Classical and Quantum\\Channel Discrimination}

\author{Theshani~Nuradha}
\email{theshani.gallage@gmail.com, nuradha@illinois.edu}
\affiliation{School of Electrical and Computer Engineering, Cornell University, Ithaca, New York 14850, USA}
\affiliation{Department of Mathematics \& \\ 
Illinois Quantum Information Science and Technology (IQUIST) Center, University of Illinois Urbana-Champaign, Urbana, IL 61801, USA}

\author{Mark~M.~Wilde}
\email{wilde@cornell.edu}
\affiliation{School of Electrical and Computer Engineering, Cornell University, Ithaca, New York 14850, USA}

\begin{abstract}
    Quantum channel discrimination has been studied from an information-theoretic perspective, wherein one is interested in the optimal decay rate of error probabilities as a function of the number of unknown channel accesses. In this paper, we study the query complexity of quantum channel discrimination, wherein the goal is to determine the minimum number of channel uses needed to reach a desired error probability. To this end, we show that the query complexity of binary channel discrimination depends logarithmically on the inverse error probability and inversely on the negative logarithm of the (geometric and Holevo) channel fidelity. As special cases of these findings, we precisely characterize the query complexity of discriminating two classical channels and two classical--quantum channels. Furthermore, by obtaining an optimal characterization of the sample complexity of quantum hypothesis testing  when the error probability does not exceed a fixed threshold, we provide a more precise characterization of query complexity under a similar error probability threshold constraint. We also provide lower and upper bounds on
the query complexity of binary asymmetric channel discrimination and multiple quantum channel discrimination. For the former, the query complexity depends on the geometric R\'enyi and Petz--R\'enyi channel divergences, while for the latter, it depends on the negative logarithm of the (geometric and Uhlmann) channel fidelity. For multiple channel discrimination, the upper bound scales as the logarithm of the number of channels. 
\end{abstract}

\maketitle 

\tableofcontents

\section{Introduction}

Quantum hypothesis testing has been studied in early works of quantum information theory~\cite{Hel67,Hel69,Hol73}. 
In quantum state discrimination, the hypotheses are modelled as quantum states, and the task is to devise a measurement
strategy that maximizes the probability of correctly identifying the state of a given quantum system. In the  basic version of this problem, a quantum system is prepared in one of two possible states, denoted by $\rho$ and $\sigma$,  and it is the goal of the distinguisher, who does not know \textit{a priori} which state was prepared, to determine the identity of the unknown state. 
 In the case that $n\in \mathbb{N}$ copies (or samples) are provided, the states are then denoted by $\rho^{\otimes n}$ and $\sigma^{\otimes n}$. Having extra samples for the discrimination task 
 decreases the probability of making an error in the process.  It is well known that the error probability decreases exponentially fast in the number of samples, provided that $\rho \neq \sigma$~\cite{HP91,ON00,ACM+07,NS09,Nag06,Hay07}. Quantum hypothesis testing can be generalized to the setting in which there are multiple hypotheses~\cite{YKL75}, and it is also known in this case that the error probability generally decays exponentially fast with $n$~\cite{Li16}.

Quantum channel discrimination is a generalization of quantum hypothesis testing and involves determining the identity of an unknown channel that is queried. Indeed, the task is to query the unknown quantum channel several times ($n$ times) to determine which channel was applied. The asymptotic regime of quantum channel discrimination (i.e., $n \to \infty$) has been extensively studied in~\cite{Hayashi_9_channel_discr,Cooney2016,wilde2020amortized,WW19,FFRS20,bergh2022parallelization,SHW22,bergh2023infinite,HuangBosonciDephasing24}. Additionally, Ref.~\cite{harrow2010adaptive} studied various settings in which a finite number of channel uses are sufficient to distinguish two channels perfectly using adaptive strategies. More broadly, there has been significant interest in the topic of quantum channel discrimination \cite{Memory_ChannelDiscriminationPRL,Duan09,PW09,MPW10,Puzzuoli2017,KW21}.

Recently, Ref.~\cite{cheng2024invitation} studied the non-asymptotic regime of quantum hypothesis testing by considering its sample complexity. Here, sample complexity refers to the minimum number of samples of the unknown state required to achieve a fixed non-zero error probability. Furthermore, Ref.~\cite{cheng2024invitation} also explored several applications of the sample complexity of quantum hypothesis testing in quantum algorithms for simulation and unstructured
search, and quantum learning and classification. This non-asymptotic contribution also provides a foundation for studying information-constrained settings of quantum hypothesis testing. These constraints include ensuring privacy~\cite{nuradha_privateQHT25,cheng_privateQHT24}, where privacy is quantified by quantum local differential privacy~\cite{hirche2023quantum, nuradha2023quantum}, and access to noisy quantum states instead of noiseless states~\cite{george2025quantum}.

 In this work, we extend the study of quantum hypothesis testing to the more general setting of channel discrimination. 
In particular, we consider how many queries of an unknown channel are required to achieve a fixed non-zero error probability for discriminating two or more channels. We use the term ``query complexity'' to refer to the minimum number of channel uses in this scenario, as this terminology is common in theoretical computer science and is used for similar settings including unitary operator and gate discrimination \cite{huang2022query,kawachi2019quantum,chiribella2013query,rossi2022quantum,ito2021lower}. Similar to the sample complexity of quantum hypothesis testing, we suspect that our findings will be useful in applications related to quantum learning theory, quantum computing, and quantum algorithms, in which one has to distinguish several channels. 

\bigskip
\textbf{Contributions:}
In this paper, we provide a comprehensive analysis of the query complexity of classical and quantum channel discrimination. First, we define the query complexity of three quantum channel discrimination settings, including symmetric binary (Definition~\ref{def:binary_symmetric_C}), asymmetric binary (Definition~\ref{def:binary_asymmetric_C}), and symmetric multiple channel discrimination (Definition~\ref{def:M-ary_C}), while providing equivalent expressions in Remark~\ref{rem:equivalent_QC}. Before proceeding to our query complexity results, we improve a sample complexity lower bound from \cite{cheng2024invitation}, and, by building on the recent finding in \cite[Theorem~2]{kazemi2025sample}, we obtain an optimal characterization of the sample complexity of binary quantum hypothesis testing under a threshold constraint on the error probability (see  Theorem~\ref{thm:alterantive_SC_binary}).

Next, we study the query complexity of symmetric binary channel discrimination. To this end, Theorem~\ref{theorem:binary_symmetric_C} states that the query complexity scales logarithmically in the inverse error probability $\varepsilon \in (0,1)$ and inversely in the negative logarithm of the geometric channel fidelity (lower bound) and Holevo channel fidelity (upper bound). Moreover, Corollary~\ref{Cor:classical_channels} provides a tight characterization of the query complexity of discriminating two classical channels $\cN$ and $\cM$, and Corollary~\ref{Cor:query_CQ_1} provides a similar tight characterization for classical--quantum channels. 

By building on the aforementioned optimal characterization of sample complexity, we then obtain a more precise characterization of the query complexity of binary channel discrimination in terms of the prior probabilities, fixed error probability, and a channel divergence (see Theorem~\ref{prop:alternative_query_C}). For the special case of two classical channels and two classical--quantum channels, the upper and lower bounds differ only by a constant and are thus optimal, in contrast to the bounds presented in Theorem~\ref{theorem:binary_symmetric_C}; see Corollary~\ref{cor:precise-classical} and Corollary~\ref{Cor:Q_C_CQ_2}, respectively.

Furthermore, we characterize the query complexity of asymmetric binary channel discrimination with the use of the geometric R\'enyi and Petz--R\'enyi channel divergences (see Theorem~\ref{thm:binary_asymmetric_C}). As the third setting, we establish bounds on the query complexity of  symmetric multiple channel discrimination (see Theorem~\ref{theorem:sc_M-ary_C}). Here, we find that the query complexity scales with the negative logarithm of the channel fidelities and the upper bound scales as $O\!\left(\ln(M)\right)$, where $M$ is the number of distinct channels to be discriminated. 

\medskip
\textit{Note: After the completion of the first arXiv post of our work \cite{nuradha2025querycomplexityclassicalquantum}, we came across the independent work~\cite{Li_25_query_channel}, which also studies the non-asymptotic regime of quantum channel discrimination.} 

\section{Preliminaries and Notations}
In this section, we establish some notation and recall various quantities of interest and relations used in our paper.

Let $\mathbb{N}\coloneqq \{1,2,\ldots\}$.
Throughout our paper, we let $\rho$ and $\sigma$ denote quantum states, which are positive semi-definite operators acting on a separable Hilbert space and with trace equal to one. Let $\cN_{A \to B}$ and $\cM_{A \to B}$ be quantum channels, which are completely positive trace-preserving maps. 
Let  $I$  denote the identity operator.
For every operator $A$, we define the trace norm as 
\begin{align} \label{eq:Schatten}
	\left\|A\right\|_1 \coloneqq  \Tr\!\left[ |A| \right].
\end{align}
where $|A| \coloneqq \left( A^\dagger A \right)^{1/2}$.

Let $f(x)$ and $g(x)$ be functions taking on non-negative values. We write $f(x) = \Theta\!\left(g(x) \right)$
if there exist constants $c_1, c_2 > 0$, and $x_0 \geq 0$ such that 
$ c_1 g(x) \leq f(x) \leq c_2 g(x)$ for all $x \geq x_0$.
Intuitively, $f(x)$ grows at the same rate as $g(x)$, up to constant factors, for sufficiently large $x$. In our paper, our findings are in fact slightly stronger than indicated by the $\Theta$ notation; that is, for every instance of the notation $f(x) = \Theta\!\left(g(x) \right)$ that appears in our paper, the following statement actually holds: for all $x>0$, there exist constants $c_1, c_2 > 0$ such that $ c_1 g(x) \leq f(x) \leq c_2 g(x)$ (i.e., the statement holds for all $x>0$ and not simply for all $x\geq x_0$ where $x_0>0$).

\subsection{Quantum Information-Theoretic Quantities}

A divergence $\bm{D}(\rho\|\sigma)$ is a function of two quantum states $\rho$ and $\sigma$ that obeys the data-processing inequality; i.e., the following holds for all states $\rho$ and $\sigma$ and every channel $\mathcal{N}$:
\begin{equation}
    \bm{D}(\rho\|\sigma) \geq \bm{D}(\mathcal{N}(\rho)\|\mathcal{N}(\sigma)).
\end{equation}
\begin{definition}[Generalized Divergences]\label{def:Gene_div}Let $\rho$ and $\sigma$ be quantum states.
    \begin{enumerate}
    \item The normalized trace distance is defined as 
		\begin{align}
			 \frac12 \left\|\rho-\sigma \right\|_1.
		\end{align}
        \item The Petz--R\'enyi divergence of order $\alpha \in (0,1) \cup (1,\infty)$ is defined as~\cite{Pet86}
        \begin{align}
            D_\alpha(\rho\|\sigma) & \coloneqq \frac{1}{\alpha-1} \ln Q_\alpha(\rho\|\sigma) \label{eq:petz-renyi-div} \\
            {\hbox{where}} \quad
            Q_\alpha(A\|B) & \coloneqq \lim_{\varepsilon \to 0^+} \Tr[A^\alpha (B+\varepsilon I)^{1-\alpha}], \qquad \forall A,B\geq 0.
        \end{align}
        It obeys the data-processing inequality for all $\alpha \in (0,1)\cup(1,2]$~\cite{Pet86}. 
        Note also that $D_{1/2}(\rho\|\sigma) = -\ln F_{\mathrm{H}}(\rho,\sigma)$, where the Holevo fidelity  is defined as~\cite{Holevo1972fid}
		\begin{align} \label{eq:Holevo_fidelity}
			F_\mathrm{H}(\rho,\sigma) \coloneqq  \left(\Tr\!\left[\sqrt{\rho} \sqrt{\sigma} \right]\right)^2.
		\end{align}

        \item The sandwiched R\'enyi divergence of order $\alpha \in (0,1) \cup (1,\infty)$ is defined as~\cite{MDS+13,wilde_strong_2014}
        \begin{align}
            \widetilde{D}_\alpha(\rho\|\sigma) & \coloneqq \frac{1}{\alpha-1} \ln \widetilde{Q}_\alpha(\rho\|\sigma) \label{eq:sandwiched-renyi-div} \\
            {\hbox{where}} \quad \widetilde{Q}_\alpha(A\|B) & \coloneqq \lim_{\varepsilon\to 0^+} \Tr[\left(A^{\frac{1}{2}} (B+\varepsilon I)^{\frac{1-\alpha}{\alpha}} A^{\frac{1}{2}} \right)^\alpha], \qquad \forall A,B\geq 0.
        \end{align}
        It obeys the data-processing inequality for all $\alpha \in [1/2,1)\cup(1,\infty)$~\cite{FL13} (see also \cite{Wilde_2018}). Note also that $\widetilde{D}_{1/2}(\rho\|\sigma) = -\ln F(\rho,\sigma)$, where the quantum fidelity is defined as~\cite{Uhl76}
		\begin{align} \label{eq:fidelity}
			F(\rho,\sigma) \coloneqq  \left\|\sqrt{\rho} \sqrt{\sigma} \right\|_1^2.
		\end{align}
  The Bures distance is defined as~\cite{Hel67b, Bur69}
		\begin{align}
			d_\mathrm{B}(\rho,\sigma) 
			&\coloneqq \sqrt{2\left(1 - \sqrt{F}(\rho,\sigma)\right)}.
		\end{align}

        \item The geometric  R\'enyi divergence of order $\alpha \in (0,1) \cup (1,\infty)$ is defined as \cite{matsumoto2015new} 
         \begin{align}
            \widehat{D}_\alpha(\rho\|\sigma) & \coloneqq \frac{1}{\alpha-1} \ln \widehat{Q}_\alpha(\rho\|\sigma) \label{eq:geometric-renyi-div} \\
            {\hbox{where}} \quad \widehat{Q}_\alpha(A\|B) & \coloneqq \lim_{\varepsilon\to 0^+} \Tr[ (B +\varepsilon I) \left( (B+\varepsilon I)^{-\frac{1}{2}} A (B+\varepsilon I)^{-\frac{1}{2}} \right)^\alpha], \quad \forall A,B\geq 0.
        \end{align}
        It obeys the data-processing inequality for all $\alpha \in (0,1)\cup(1,2]$ \cite[Theorem~7.45]{KW20}. For $\alpha=1/2$, we have 
        \begin{equation}
            \widehat{D}_{\frac{1}{2}}(\rho\|\sigma) =- \ln \widehat{F}(\rho, \sigma), 
        \end{equation}
      where the geometric fidelity \cite{matsumoto2010reverse} is defined as
\begin{equation}
\label{eq:intro-geo-fid-def}
     \widehat{F}(\rho,\sigma)\coloneqq \inf_{\varepsilon >0 } \left(\Tr[ \rho(\varepsilon) \# \sigma(\varepsilon)] \right)^2.
\end{equation}
Here $A \# B$ denotes the matrix geometric mean of the positive definite operators $A$ and~$B$:
\begin{equation}
    A \# B \coloneqq A^{1/2}\left( A^{-1/2} B A^{-1/2} \right)^{1/2} A^{1/2}, \label{eq:geometric-mean}
\end{equation}
$\rho(\varepsilon) \coloneqq \rho + \varepsilon I$,
and $\sigma(\varepsilon) \coloneqq \sigma + \varepsilon I$.
Using the geometric fidelity, we also define
        \begin{align} \label{eq:d_B_hat}
			d_{\widehat{F}}(\rho,\sigma) 
			&\coloneqq \sqrt{2\left(1 - \sqrt{\widehat{F}}(\rho,\sigma)\right)}.
		\end{align}
 Note that $d_{\widehat{F}}$ is not a distance metric because it does not satisfy the triangular inequality (see \cite[page~1787]{bhatia2019matrix} for a counterexample).

          \item The quantum relative entropy is defined as~\cite{Ume62}
        \begin{equation}
            D(\rho\|\sigma) \coloneqq \lim_{\varepsilon \to 0^+} \Tr[\rho(\ln \rho - \ln (\sigma + \varepsilon I))],
        \end{equation}
        and it obeys the data-processing inequality~\cite{lindblad1975completely}.
    \end{enumerate}
\end{definition}

Generalized divergences can be extended to quantum channels $\cN_{A \to B}$ and $\cM_{A \to B}$ as follows \cite[Definition~II.2]{LKDW18}:
\begin{equation} \label{eq:channel_divergences}
    \bm{D}(\cN \Vert \cM) \coloneqq \sup_{\rho_{RA}} \bm{D}(\cN_{A \to B}(\rho_{RA}) \Vert \cM_{A \to B}(\rho_{RA}) ).
\end{equation}
Then, we can also define the channel variants of all the generalized divergences presented in Definition~\ref{def:Gene_div} (replacing $\bm{D}$ with the following: $D_\alpha, \widetilde{D}_\alpha, \widehat{D}_\alpha, d_B, d_{\widehat{F}})$. In the same spirit, channel fidelities (denoted as $\bm{F}$ to include $F_H, F , \widehat{F}$) are defined as 
\begin{equation}\label{eq:channel_fidelities}
    \bm{F}(\cN, \cM) \coloneqq \inf_{\rho_{RA}} \bm{F}(\cN_{A \to B}(\rho_{RA}) ,\cM_{A \to B}(\rho_{RA}) ).
\end{equation}
As shown in \cite[Proposition~55]{katariya2021geometric}, the channel fidelity based on $F$ can be efficiently computed by means of a semi-definite program (SDP), and a similar statement can be made for $\widehat{F}$ \cite[Proposition~48]{katariya2021geometric}. Precisely, the SDP is given by
\begin{equation}
    \left[\widehat{F}(\cN,\cM)\right]^{1/2} =\sup_{\lambda \geq 0, W_{RB} \in \operatorname{Herm}} \left\{ \lambda : \lambda I_R \leq \Tr_B[W_{RB}] 
     ,   \begin{bmatrix}
        \Gamma^{\cN}_{RB} & W_{RB} \\
        W_{RB} & \Gamma^{\cM}_{RB}
        \end{bmatrix} \geq 0\right\},
\end{equation}
where $\Gamma^{\cN}_{RB} $ and $\Gamma^{\cM}_{RB}$ are the Choi operators for $\cN $ and $\cM$, respectively, defined as $\Gamma^{\cN}_{RB} \coloneqq \sum_{i,j} |i\rangle\!\langle i|\otimes \mathcal{N}(|i\rangle\!\langle j|)$.

We make use of the following inequalities throughout our work: Let $\cN_{A \to B}$ and $\cM_{A \to B}$ denote quantum channels, and  let $\rho_{RA}$ and $\sigma_{RA}$ denote quantum states.
\begin{enumerate}
    \item Chain rule for geometric R\'enyi divergence (for $\alpha \in (1,2]$ \cite[Theorem~3.4]{fang2021geometric} and for $\alpha \in (0,1)$ \cite[Proposition~45]{katariya2021geometric}):
    \begin{equation}
    \label{eq:chain_rule_geometric_renyi}
        \widehat{D}_\alpha\!\left( \cN_{A \to B}(\rho_{RA}) \Vert \cM_{A \to B} (\sigma_{RA})\right) \leq \widehat{D}_\alpha(\cN \Vert \cM) + \widehat{D}_\alpha(\rho_{RA} \Vert \sigma_{RA}).
    \end{equation}
    \item Since $\widehat{D}_{\frac{1}{2}}(\rho\|\sigma) =- \ln \widehat{F}(\rho, \sigma)$, the following is a rewriting of \eqref{eq:chain_rule_geometric_renyi}:
    \begin{equation}\label{eq:chain_rule_fide}
        \widehat{F}\!\left( \cN_{A \to B}(\rho_{RA}), \cM_{A \to B} (\sigma_{RA})\right) \geq \widehat{F}(\cN, \cM) \ \widehat{F}(\rho_{RA},  \sigma_{RA}).
    \end{equation}
    \item Let $n \in \mathbb{N}$. Then, 
    \begin{equation}
        F\!\left(  \cN^{\otimes n}, \cM^{\otimes n} \right) \leq \left[  F(\cN, \cM) \right]^n. \label{eq:fidelity_channel_ineq_n}
    \end{equation}
    This follows because
    \begin{align}
        F\!\left(  \cN^{\otimes n}, \cM^{\otimes n} \right)  &= \inf_{\rho_{R_n A_n}} F\!\left(  \cN^{\otimes n}(\rho_{R_n A_n}), \cM^{\otimes n}(\rho_{R_n A_n}) \right) \\
        & \leq \inf_{\rho_{RA}^{\otimes n}} F\!\left(  \cN^{\otimes n}(\rho_{RA}^{\otimes n}), \cM^{\otimes n}(\rho_{RA}^{\otimes n}) \right) \\ 
        &= \inf_{\rho_{RA}}\left[  F\!\left(\cN(\rho_{RA}), \cM(\rho_{RA}) \right) \right]^n \\ 
        &=\left[  F(\cN, \cM) \right]^n,
    \end{align}
where the first inequality holds by choosing a smaller set comprised of product states for the optimization, and the penultimate equality from the multiplicativity of fidelity.

Note that it is known from~\cite{UnitaryDistinguish_PRL} that, for every two distinct and non-orthogonal unitary channels $\mathcal{U}$ and $\mathcal{V}$, there exists a value $n\in \mathbb{N}$ such that $F\!\left(  \mathcal{U}^{\otimes n}, \mathcal{V}^{\otimes n} \right)=0$, whereas $F\!\left(  \mathcal{U}, \mathcal{V} \right)\in (0,1)$. So this represents a case in which the inequality in \eqref{eq:fidelity_channel_ineq_n} is strict. Alternatively, the inequality in~\eqref{eq:fidelity_channel_ineq_n} can be saturated for some pairs of channels, including replacer channels.
\end{enumerate}

\subsection{Setup} \label{Sec:Setup}

 \begin{figure}
        \centering
        \includegraphics[width=\linewidth]{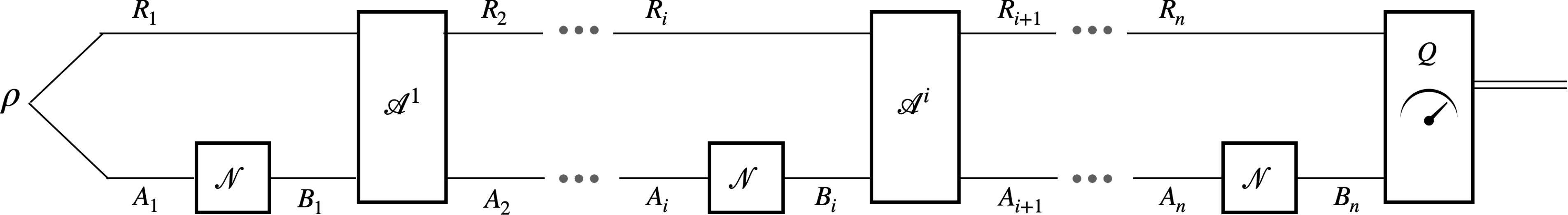}  
        \caption{Channel discrimination protocol when the unknown channel is $\cN$. In the case that the unknown channel is $\cM$, we replace $\cN$ in the figure by $\cM$.}
        \label{fig:channelDiscriminationSetup}
    \end{figure}

Given two quantum channels $\mathcal{N}_{A \to B}$ and $\mathcal{M}_{A \to B}$ acting on an input system $A$ and an output system $B$, the most general adaptive strategy for discriminating these two channels is as follows (see Fig.~\ref{fig:channelDiscriminationSetup} for a depiction). First we prepare an arbitrary input state $\rho_{R_1A_1}= \tau_{R_1A_1}$, where $R_1$ is a reference system (ancillary register). The $i$th use of the respective channel takes a register $A_i$ as input and gives a register $B_i$ as output. After each use of the channel $\cN$ or $\cM$, we apply an (adaptive) channel $\cA_{R_i B_i \to R_{i+1} A_{i+1}}^{(i)}$ to the registers $R_i$ and $B_i$ such that the following holds (depending on the channel applied): for $1\leq i \leq n$,
\begin{align}
    \rho_{R_{i} B_{i}} & \coloneqq \cN_{A_i \to B_i} (\rho_{R_i A_i}), \\ 
    \tau_{R_{i} B_{i}} & \coloneqq  \cM_{A_i \to B_i} (\tau_{R_i A_i}),
\end{align}
and 
\begin{align}
    \rho_{R_{i+1} A_{i+1}} & \coloneqq (\cA_{R_i B_i \to R_{i+1} A_{i+1}}^{(i)} \circ \cN_{A_i \to B_i}) (\rho_{R_i A_i}), \\ 
    \tau_{R_{i+1} A_{i+1}} & \coloneqq (\cA_{R_i B_i \to R_{i+1} A_{i+1}}^{(i)} \circ \cM_{A_i \to B_i}) (\tau_{R_i A_i}),
\end{align}
for every $1 \leq i < n$ on the left-hand side, where $n \in \mathbb{N}$. Finally, a quantum measurement $( Q_{R_n B_n} , I_{R_n B_n}- Q_{R_n B_n} )$ is applied on the systems $R_n B_n$ to determine which channel was queried. The outcome $Q$ corresponds to guessing the channel $\cN$ at the end of the protocol, while the outcome $I-Q$ corresponds to guessing the channel $\cM$. The probability of guessing $\cN$ when $\cM$ is applied is given by $\Tr\!\left[ Q_{R_n B_n} \tau_{R_n B_n}\right]$ (type II error probability), while the probability of guessing $\cM$ when $\cN$ is applied is given by $\Tr\!\left[ (I_{R_n B_n} -Q_{R_n B_n} )\rho_{R_n B_n}\right]$ (type I error probability).

In what follows, we use the notation $\{Q, \cA\}$ to identify a strategy of using the tuple of channels, $\left(\cA_{R_i B_i \to R_{i+1} A_{i+1}}^{(i)} \right)_i$, and a final measurement $( Q_{R_n B_n} , I_{R_n B_n}- Q_{R_n B_n} )$ as described above. In this setup, we also suppose that we have access to $n \in \mathbb{N}$ queries of the channel $\cN$ or $\cM$.

Finally, consider that a product/parallel strategy consists of preparing $n$ copies of a bipartite state, followed by the unknown channel $\cN$ or $\cM$ acting on one share of each state, and then applying a collective measurement on all of the systems. Such a strategy is actually a special case of the general strategy described above~\cite{Memory_ChannelDiscriminationPRL}. See also~\cite[Remark~1]{wilde2020amortized}.

\subsubsection{Symmetric Binary Channel Discrimination}

\label{subS:symmetric_binary}

In the setting of symmetric binary channel discrimination, we suppose that there is a prior probability $p \in (0,1)$ associated with the use of the channel $\cN$, and there is a prior probability $q \equiv 1-p$ associated with using the channel $\cM$. Indeed, the unknown channel is selected by flipping a coin, with the probability of heads being $p$ and the probability of tails being~$q$. If the outcome of the coin flip is heads, then $n$ channel uses of $\cN$ are applied, and if the outcome is tails, then $n$ channel uses of $\cM$ are applied, along with the adaptive channels $\left(\cA_{R_i B_i \to R_{i+1} A_{i+1}}^{(i)} \right)_i$ for $1 \leq i <n$ and the input state $\rho_{R_1 A_1}= \tau_{R_1A_1}$ in both cases.  Thus, the expected error probability in this experiment is as follows:
\begin{align}
    p_{e,\{Q,\cA\}, \rho_{R_1A_1}}(p, \cN,q,\cM, n) & \coloneqq   p \Tr\!\left[ (I_{R_n B_n} -Q_{R_n B_n} )\rho_{R_n B_n}\right] + q \Tr\!\left[ Q_{R_n B_n} \tau_{R_n B_n}\right]
    \label{eq:err-prob-fixed-meas_C}.
\end{align}

Given $p$, the distinguisher can minimize the error-probability expression in~\eqref{eq:err-prob-fixed-meas_C} over all adaptive strategies $\{Q, \cA\}$ and the initial state $\rho_{R_1A_1}$.
The optimal error probability $p_e(p, \cN,q,\cM,n)$ of hypothesis testing is as follows:
\begin{align}
	p_e(p, \cN,q,\cM,n) &\coloneqq  \inf_{\{Q, \cA\},\rho_{R_1A_1}}  p_{e,\{Q,\cA\}, \rho_{R_1A_1}}(p, \cN,q,\cM, n) \\
 \label{eq:Helstrom1_C}
 & =\inf_{\{Q, \cA\},\rho_{R_1A_1}} \left\{  p \Tr\!\left[ (I_{R_n B_n} -Q_{R_n B_n} )\rho_{R_n B_n}\right] + q \Tr\!\left[ Q_{R_n B_n} \tau_{R_n B_n}\right]  \right\}
	\\
	\label{eq:Helstrom2_C}	
	&= \inf_{ \cA,\rho_{R_1A_1}} \frac12 \left( 1 - \left\| p \rho_{R_n B_n} - q \tau_{R_n B_n} \right\|_1 \right), 
\end{align}
where the last equality follows by applying the Helstrom--Holevo theorem~\cite{Hel69, Hol73} for states and recalling that the optimization over all final measurements is abbreviated by $Q$.

\subsubsection{Asymmetric Binary Channel Discrimination}
In the setting of asymmetric channel discrimination, there are no prior probabilities associated with the selection of the channel $\cN$ or $\cM$---we simply assume that one of them is chosen deterministically, but the chosen channel is unknown to the distinguisher. The goal of the distinguisher is to minimize the probability of the second kind of error subject to a constraint on the probability of the first kind of error in guessing which channel is selected. Given a fixed $\varepsilon \in [0,1]$, the scenario reduces to the following optimization problem:
\begin{equation} 
\beta_{\varepsilon}( \cN^{(n)}\Vert \cM^{(n)})\coloneqq\inf_{\{ Q, \cA\},\rho_{R_1A_1}}\left\{
\begin{array}
[c]{c}
\Tr\!\left[ Q_{R_n B_n} \tau_{R_n B_n}\right]:\Tr\!\left[ (I_{R_n B_n} -Q_{R_n B_n} )\rho_{R_n B_n}\right] \leq\varepsilon
\end{array}
\right\}  .
\label{eq:beta-err-asymm_C}
\end{equation}

\subsubsection{Multiple Channel Discrimination}

Here the goal is to select one among $M$ possible channels as the guess for the unknown channel. 
Let $\mathcal{S} \coloneqq  \left( (p_m,  \cN^m)\right)_{m=1}^M $ be an ensemble of $M$ channels with prior probabilities taking values in the tuple $\left( p_m \right)_{m=1}^M$.
Without loss of generality, let us assume that $p_m>0$ for all $m\in \{1,2,\ldots, M\}$.
Fix the adaptive channels applied after each channel use as $\left(\cA_{R_i B_i \to R_{i+1} A_{i+1}}^{(i)} \right)_i$ and the input to the protocol as $\rho_{R_1A_1}$.
Let  $\rho_{R_1A_1}^m \coloneqq \rho_{R_1A_1}$ for all $1 \leq m \leq M$,  and 
\begin{equation}
   \rho_{R_{i+1} A_{i+1}}^m  \coloneqq ( \cA_{R_i B_i \to R_{i+1} A_{i+1}}^{(i)} \circ \cN^m_{A_i \to B_i}) (\rho_{R_i A_i}^m),
\end{equation}
for all $1\leq i <n$.

The minimum error probability of $M$-channel discrimination, given $n$ channel uses of the unknown channel, is as follows:
\begin{align}
	\label{eq:error_M-ary_C}
	p_e(\mathcal{S},n) 
	&\coloneqq \inf_{\{ Q, \cA\},\rho_{R_1A_1}} \sum_{m=1}^M p_m \Tr\!\left[ (I_{R_n B_n} -Q_{R_n B_n}^m )\rho_{R_n B_n}^m\right],
\end{align}
where the minimization is over all adaptive strategies denoted by $\cA$, every initial state $\rho_{R_1A_1} $, and every positive operator-valued measure (POVM) denoted as $Q$ applied at the final stage (i.e., a tuple $(Q^1_{R_n B_n},\ldots, Q^M_{R_n B_n})$ satisfying $Q^m_{R_n B_n} \geq 0$ for all $m\in \{1,\ldots,M\}$ and $\sum_{m=1}^M Q^m_{R_n B_n} = I_{R_n B_n}$).

\section{Improved Bounds for Sample Complexity}

In this section, we provide improved characterizations of the sample complexity of binary quantum hypothesis testing, when compared to some from \cite{cheng2024invitation}. For our first result here, we provide an improved lower bound when compared to one of the lower bounds from \cite[Theorem~II.7]{cheng2024invitation}. For our second result, we draw inspiration from~\cite[Theorem~2]{kazemi2025sample}, which has recently improved the sample complexity bounds for classical binary hypothesis testing. This latter result leads to an optimal characterization of sample complexity under a threshold constraint on the error probability. 

First, let us recall symmetric binary quantum hypothesis testing of states: Suppose that there are two states $\rho$ and $\sigma$, and  $\rho^{\otimes n}$ is selected with probability $p\in (0,1)$ and $\sigma^{\otimes n}$ is selected with probability $q\equiv 1-p$. The sample complexity of this setting is defined in~\cite[Definition~II.3, Remark~I.1] {cheng2024invitation} as follows: 
\begin{equation}\label{eq:SC_states_binary}
    n^*(p,\rho,q,\sigma,\varepsilon) \coloneqq \inf \left\{ n\in \mathbb{N}: \frac{1}{2}\left(  1-\left\Vert p\rho^{\otimes n}-q\sigma^{\otimes
n}\right\Vert _{1}\right) \leq \varepsilon \right\}.
\end{equation}

First, we establish a lower bound that improves upon one of the lower bounds from \cite[Theorem~7]{cheng2024invitation}. Additionally, by comparing with \cite[Theorem~II.6]{cheng2024invitation} observe that this lower bound is tight when both $\rho$ and $\sigma$ are pure states.

\begin{theorem} 
\label{theorem:binary_symmetric_improved_lower_bnd}
	Let $p$, $q$, $\varepsilon$, $\rho$, and $\sigma$ be as stated in \eqref{eq:SC_states_binary}.
	Then the following lower bound on sample complexity holds:
	\begin{align}
		\label{eq:binary_symmetric3}
		 \frac{\ln\!\left(\frac{pq}{\varepsilon\left(1-\varepsilon\right)}\right) }{ -\ln F(\rho,\sigma) }  \leq n^*(p,\rho,q,\sigma, \varepsilon).		
	\end{align}
\end{theorem}

\begin{proof}
See Appendix~\ref{app:proof-improved-samp-comp-lower-bnd}.
\end{proof}

Next, we provide alternative sample complexity bounds, which lead to a tighter characterization in terms of the prior probabilities $p$ and $q$, when compared to~\cite[Theorem~II.7 \& Corollary~II.8]{cheng2024invitation}. In fact, the resulting lower and upper bounds differ only by a constant and are thus optimal, as indicated in Theorem~\ref{thm:alterantive_SC_binary} below.

Towards that, we define the following quantity for $p \in (0,1/2]$, $q \equiv 1-p $, and $\varepsilon \in( 0,p)$:
\begin{equation}\label{eq:lambda_*}
    \lambda_* \coloneqq \frac{\ln\!\left( \frac{q}{\varepsilon}\right)}{\ln\!\left( \frac{q}{\varepsilon}\right)+\ln\!\left( \frac{p}{\varepsilon}\right)}.
\end{equation}
Observe that $\lambda_* \in [1/2,1)$ for $p \in(0,1/2]$. Also, define the following quantity for $s\in (0,1)$:
\begin{equation}\label{eq:Q_s}
    Q_s(\rho \Vert \sigma) \coloneqq \Tr\!\left[ \rho^s \sigma^{1-s} \right].
\end{equation}
The quantity $Q_s$ is multiplicative in the following sense, for states $\rho_1$, $\rho_2$, $\sigma_1$, and $\sigma_2$ and $s\in(0,1)$:   
\begin{equation} \label{eq:Q_s_multiplicativity}
     Q_s(\rho_1 \otimes \rho_2 \Vert \sigma_1 \otimes \sigma_2) = Q_s(\rho_1 \Vert \sigma_1) \, Q_s(\rho_2 \Vert \sigma_2). 
\end{equation}

\begin{theorem} \label{thm:alterantive_SC_binary}
   Let $p \in (0,1/2]$, $q \equiv 1-p$, $\varepsilon \in (0, p/64)$, $\rho$, and $\sigma$ be as stated in~\eqref{eq:SC_states_binary}.  Then, we have that 
\begin{equation}
\left \lceil \frac{1}{2} \lambda_* \frac{ \ln\!\left( \frac{p}{\varepsilon} \right)}{-\ln Q_{\lambda^*}(\rho \Vert \sigma)} \right \rceil   \leq  n^*(p,\rho,q,\sigma,\varepsilon) \leq \left \lceil 2 \lambda_* \frac{ \ln\!\left( \frac{p}{\varepsilon} \right)}{-\ln Q_{\lambda^*}(\rho \Vert \sigma)} \right \rceil, 
\end{equation}
where $\lambda_*$ and $Q_{\lambda_*}(\cdot \Vert \cdot)$ are defined in~\eqref{eq:lambda_*} and~\eqref{eq:Q_s}, respectively. 

Consequently, we have that 
\begin{equation}
   n^*(p,\rho,q,\sigma,\varepsilon) = \Theta\! \left (\lambda_* \frac{ \ln\!\left( \frac{p}{\varepsilon} \right)}{-\ln Q_{\lambda^*}(\rho \Vert \sigma)} \right ).
\end{equation}
\end{theorem}
\begin{proof}
To arrive at the upper bound, for all $\lambda \in (0,1)$, consider that
\begin{align}
    \frac{1}{2} \left(1- \left\| p \rho^{\otimes n} - q \sigma^{\otimes n} \right\|_1  \right) &  \leq  p^\lambda q ^{1-\lambda}\Tr\!\left[ \left(\rho^{\otimes n}\right) ^\lambda \left(\sigma^{\otimes n} \right)^{1-\lambda} \right] \\ 
    &= p^\lambda q^{1-\lambda} \left( Q_\lambda(\rho \Vert \sigma) \right)^n, \label{eq:upper_bound_mid}
\end{align}
where the first inequality follows by setting $A= p\rho^{\otimes n}$ and $B=q \sigma^{\otimes n}$ in~\cite[Theorem~1]{ACM+07} (i.e., for operators $A, B \geq 0$ and $0\leq s \leq 1$, we have $\Tr\!\left[ A^s B^{1-s}\right] \geq \Tr\!\left[A+B -|A-B| \right]/2 $), and the equality from~\eqref{eq:Q_s_multiplicativity}.
By choosing $\lambda= \lambda_*$ and $\varepsilon \in (0,p)$, we have that 
\begin{align}
   \frac{ p^{\lambda_*} q^{\lambda_*}}{\varepsilon} &= \left(\frac{p}{\varepsilon}\right)^{\lambda_*} \left(\frac{q}{\varepsilon}\right)^{1-\lambda_*} \\
   &= \left(\frac{p}{\varepsilon}\right)^{2\lambda_*}, \label{eq:p_q_lambda_rel}
\end{align}
where the last equality follows from~\eqref{eq:lambda_*}.

With that, together with~\eqref{eq:upper_bound_mid}, we arrive at 
\begin{equation}
    \frac{1}{2} \left(1- \left\| p \rho^{\otimes n} - q \sigma^{\otimes n} \right\|_1  \right) \leq \varepsilon \left(\frac{p}{\varepsilon}\right)^{2\lambda_*} \left( Q_{\lambda_*}(\rho \Vert \sigma) \right)^n.
\end{equation}
By choosing 
\begin{equation}
    n =\left \lceil 2 \lambda_* \frac{ \ln\!\left( \frac{p}{\varepsilon} \right)}{-\ln Q_{\lambda^*}(\rho \Vert \sigma)} \right \rceil, 
\end{equation}
we conclude that the average error probability satisfies $ \frac{1}{2} \left(1- \left\| p \rho^{\otimes n} - q \sigma^{\otimes n} \right\|_1  \right)  \leq \varepsilon$, completing the proof of the upper bound.

\medskip
For the lower bound, consider the regime $\varepsilon \in (0,p/64)$. Let 
\begin{equation}
    p_e(p,\rho,q,\sigma,n) \coloneqq  \frac{1}{2} \left(1- \left\| p \rho^{\otimes n} - q \sigma^{\otimes n} \right\|_1  \right).
\end{equation}
Let spectral decompositions of the states $\rho$ and $\sigma$ be given by 
\begin{equation}
    \rho = \sum_{i} \mu_i | \phi_i\rangle\!\langle \phi_i|, \quad \quad  \sigma = \sum_{j} \nu_j | \psi_j\rangle\!\langle \psi_j|.
\end{equation}
Then, as in~\cite{ANS+08}, define the following induced probability distributions:
\begin{equation}
    t_{i,j} \coloneqq \mu_i | \langle \phi_i| \psi_j\rangle|^2, \quad \quad r_{i,j} \coloneqq \nu_j | \langle \phi_i| \psi_j\rangle|^2.
\end{equation}

By~\cite[Proposition~1]{ANS+08}, we have that 
\begin{equation} \label{eq:NS_Q_S_Eq}
    Q_\lambda(\rho \Vert \sigma) = \sum_{i,j} t_{i,j}^s r_{i,j}^{1-s}.
\end{equation}
Then, by~\cite[Proposition~2, Eq.~(20)]{ANS+08}, we also have that
\begin{equation}
    p_e(p,\rho,q,\sigma,n) \geq \frac{1}{2}p_e^c (p,t,q,r,n), \label{eq:LB_r1}
\end{equation}
where $p_e^c$ corresponds to the error of binary hypothesis testing in the classical setting of distinguishing probability distributions $t$ and $r$ with $n$ samples of the data; i.e., 
\begin{equation}
    p_e^c (p,t,q,r,n) \coloneqq   \sum_{i,j} \min \! \left\{ p( t_{i,j})^n, q( r_{i,j})^n \right \}.  
\end{equation}
By applying~\cite[Theorem~1]{kazemi2025sample}, we also have that for $\lambda_{*}$ set as in~\eqref{eq:lambda_*},
\begin{align}
    p_e^c (p,t,q,r,n) &\geq \frac{1}{4}  p^{\lambda_*} q^{1- \lambda_*}  \sum_{i,j} (t_{i,j})^{2n\lambda_*} (r_{i,j})^{2n(1-\lambda_*)} \\
    &=\frac{1}{4}  p^{\lambda_*} q^{1- \lambda_*} \left(\sum_{i,j} t_{i,j}^{\lambda_*} r_{i,j}^{1-\lambda_*} \right)^{2n} \\
    &= \frac{1}{4}  p^{\lambda_*} q^{1- \lambda_*} \left( Q_{\lambda_*}(\rho \Vert \sigma)\right)^{2n}, \label{eq:LB_r2}
\end{align}
where the last equality follows by~\eqref{eq:NS_Q_S_Eq}.
Also, we have that 
\begin{align}
    \left(\frac{ p}{\varepsilon} \right)^{\lambda_*} \geq  \left(\frac{ p}{\varepsilon} \right)^{1/2} \geq 8,
\end{align}
because $\lambda_* \geq 1/2 $ and $\varepsilon \in (0,p/64)$.

With that, by combining the inequalities~\eqref{eq:LB_r1} and~\eqref{eq:LB_r2}, we arrive at
\begin{align}
    \varepsilon &\geq  p_e(p,\rho,q,\sigma,n) \\
    & \geq \frac{1}{8}  p^{\lambda_*} q^{1- \lambda_*} \left( Q_{\lambda_*}(\rho \Vert \sigma)\right)^{2n} \\
    & \geq  \left(\frac{\varepsilon} {p}\right)^{\lambda_*} p^{\lambda_*} q^{1- \lambda_*} \left( Q_{\lambda_*}(\rho \Vert \sigma)\right)^{2n}  \\ 
    &= \left(\frac{\varepsilon} {p}\right)^{\lambda_*} \left(\frac{p}{\varepsilon} \right)^{2 \lambda_*}  \varepsilon \left( Q_{\lambda_*}(\rho \Vert \sigma)\right)^{2n} \\
    & =  \left(\frac{p}{\varepsilon} \right)^{\lambda_*}  \varepsilon \left( Q_{\lambda_*}(\rho \Vert \sigma)\right)^{2n},  \label{eq:start_for_channel_proof_Q_s}
\end{align}
where the last equality follows by~\eqref{eq:p_q_lambda_rel}. This implies that the following inequality holds: 
\begin{equation}
    n \geq  \frac{1}{2} \lambda_* \frac{ \ln\!\left( \frac{p}{\varepsilon} \right)}{-\ln Q_{\lambda^*}(\rho \Vert \sigma)},
\end{equation}
which concludes the proof of the lower bound.
\end{proof}

Note that the upper bound in Theorem~\ref{thm:alterantive_SC_binary} is valid for $\varepsilon \in (0,p)$, and it is for the lower bound that we require $\varepsilon \in (0, p/64)$.

\begin{remark}[Optimality of the Sample Complexity of Quantum Hypothesis Testing]
    In~\cite[Corollary~II.8]{cheng2024invitation}, it was shown that for constant $p$ and $q$
    \begin{equation}
         n^*(p,\rho,q,\sigma,\varepsilon) = \Theta\! \left (\frac{ \ln\!\left( \frac{1}{\varepsilon} \right)}{-\ln F(\rho, \sigma)} \right ).
    \end{equation}
Furthermore, it was shown that this is not uniquely characterized by the Uhlmann fidelity $F$ or the Holevo fidelity $F_H$, and in fact can be replaced with a general class of fidelities known as $z$-fidelities (see~\cite[Eq.~(II.42) and Eq.~(II.43)]{cheng2024invitation}). 

From Theorem~\ref{thm:alterantive_SC_binary} above, we see that there is a unique characterization of sample complexity with upper and lower bounds differing only by a factor of four when $\varepsilon \in (0, p/64)$. As such, this characterization is indeed unique and also useful in practical applications where constants make a noticeable difference in the execution of the protocol, along with the challenges imposed by decoherence and noise in quantum hardware. 
\end{remark}

\section{Problem Statement for Query Complexity}

Formally, we state the  definitions of the query complexity of symmetric binary, asymmetric binary, and multiple channel discrimination in the following Definitions~\ref{def:binary_symmetric_C}, \ref{def:binary_asymmetric_C}, and~\ref{def:M-ary_C}, respectively. 
In each case, we define them by employing the error-probability metrics in~\eqref{eq:Helstrom1_C},~\eqref{eq:beta-err-asymm_C}, and~\eqref{eq:error_M-ary_C} and in such a way that the query complexity is equal to the minimum value of $n\in\mathbb{N}$ (number of channel uses) needed to get the error probability below a threshold~$\varepsilon \in [0,1]$. Our definitions for query complexity of channel discrimination are motivated by the definitions of sample complexity introduced in~\cite{cheng2024invitation} for the discrimination of two or multiple states. Also, note that for the binary channel discrimination setting, they are defined similarly in the independent work~\cite{Li_25_query_channel}.

\begin{definition}[Query Complexity of Symmetric Binary Channel Discrimination] \label{def:binary_symmetric_C}
	Let $p\in(0,1)$, $q \equiv 1-p$, and $\varepsilon \in [0,1]$, and let $\cN$ and $\cM$ be quantum channels.
	The query complexity $n^*(p,\cN,q,\cM, \varepsilon)$ of symmetric binary channel discrimination is defined as follows:
	\begin{align}
		n^*(p,\cN,q,\cM, \varepsilon)
\label{eq:def:sample_complexity_symmetric1_C}
		&\coloneqq \inf\left\{ n \in \mathbb{N} : p_e(p, \cN,q,\cM,n) \leq \varepsilon \right \} .
	\end{align}
\end{definition}

\begin{definition}[Query Complexity of Asymmetric Binary Channel Discrimination] 
\label{def:binary_asymmetric_C}
Let $\varepsilon,\delta\in\left[  0,1\right]  $, and let  $\cN$ and $\cM$ be quantum channels. The query complexity $n^{\ast}(\cN,\cM,\varepsilon,\delta)$ of asymmetric binary channel discrimination is
defined as follows:
\begin{equation}
n^{\ast}(\cN,\cM,\varepsilon,\delta)\coloneqq \inf\left\{  n\in\mathbb{N}
:\beta_{\varepsilon} \!\left(\cN^{(n)}\Vert \cM^{(n)} \right)\leq
\delta\right\}  \label{eq:asymm-beta-rewrite_C}  .
\end{equation}

\end{definition}

\begin{definition}[Query Complexity of $M$-ary Channel Discrimination] \label{def:M-ary_C}
	Let $\varepsilon\in[0,1]$, and let $\mathcal{S} \coloneqq  \left( (p_m,  \cN^m)\right)_{m=1}^M $ be an ensemble of $M$ channels.
	The query complexity $n^*(\mathcal{S},\varepsilon)$ of $M$-ary channel discrimination is defined as follows:
	\begin{align} \label{eq:def:sc_M-ary_C}
		n^*(\mathcal{S},\varepsilon) 
		& \coloneqq  \inf\left\{ n \in \mathbb{N} : p_e(\mathcal{S},n) \leq \varepsilon\right\}.
	\end{align}
\end{definition}

\begin{remark}[Equivalent Expressions for Query Complexities] \label{rem:equivalent_QC}
The query complexity \\ $n^*(p,\cN,q,\cM, \varepsilon)$ of symmetric binary quantum channel discrimination has the following equivalent expressions:
	\begin{align}
		& n^*(p,\cN,q,\cM, \varepsilon) \notag \\
& =     \inf_{\{ Q, \cA\},\rho_{R_1A_1}}
		\left\{ n\in \mathbb{N} : p \Tr\!\left[ (I_{R_n B_n} -Q_{R_n B_n} )\rho_{R_n B_n}\right] + q \Tr\!\left[ Q_{R_n B_n} \tau_{R_n B_n}\right]  \leq \varepsilon 
		\right\}
		\\
		&= \inf_{ \cA,\rho_{R_1A_1}} \left\{ n\in\mathbb{N}: \frac12 \left( 1 - \left\| p \rho_{R_n B_n} - q \tau_{R_n B_n} \right\|_1\right) \leq \varepsilon  \right\} \label{eq:def:sample_complexity_symmetric2_C}.
\end{align}
    
 By recalling  the quantity $\beta_{\varepsilon}(\cN^{(n)}\Vert\cM^{( n)})$ defined in~\eqref{eq:beta-err-asymm_C}, 
we can rewrite the sample complexity $n^{\ast}(\cN,\cM,\varepsilon
,\delta)$ of asymmetric binary quantum hypothesis testing in the following two ways:
\begin{align}
n^{\ast}(\cN,\cM,\varepsilon,\delta)  & =\inf_{\{ Q, \cA\},\rho_{R_1A_1}}\left\{
\begin{array}
[c]{c}
n\in\mathbb{N}:
\Tr\!\left[ (I_{R_n B_n} -Q_{R_n B_n} )\rho_{R_n B_n}\right] \leq\varepsilon,\\
\Tr\!\left[ Q_{R_n B_n} \tau_{R_n B_n}\right]\leq\delta
\end{array}
\right\}
\label{eq:asymm-beta-rewrite-both-errs}
\\
& =\inf\left\{  n\in\mathbb{N}:\beta_{\delta}\!\left( \cM^{(n)} \Vert
\cN^{(n)}\right)\leq\varepsilon\right\}  .\label{eq:asymm-beta-rewrite-2}
\end{align}
This follows from reasoning similar to that presented in~\cite[Section~C of supplementary material]{cheng2024invitation} because 
the expression in~\eqref{eq:asymm-beta-rewrite-both-errs} indicates that the query complexity for asymmetric binary quantum channel discrimination can be thought of as the minimum number of queries to the unknown channels needed for the type~I error probability to be below $\varepsilon$ and the type~II error probability to be below $\delta$.

Finally, the query complexity $n^*(\mathcal{S},\varepsilon)$ of $M$-ary quantum channel discrimination can be rewritten as follows:
	\begin{align}
		n^*(\mathcal{S},\varepsilon) 
  & =\inf_{\{ Q, \cA\},\rho_{R_1A_1}} 
		\left\{ n\in \mathbb{N}: \ \sum_{m=1}^M p_m \Tr\!\left[ (I_{R_n B_n} -Q_{R_n B_n}^m )\rho_{R_n B_n}^m\right] \leq \varepsilon
		\right\}.
	\end{align}
    
\end{remark}
\section{Query Complexity of Channel Discrimination}

In this section, we provide lower and upper bounds on the query complexity of quantum channel discrimination, applicable to the three main settings discussed in Section~\ref{Sec:Setup} (namely, symmetric binary channel discrimination, asymmetric binary channel discrimination, and symmetric multiple channel discrimination). 

\subsection{Symmetric Binary Channel Discrimination}

\begin{proposition}[Trivial Cases]\label{Prop:Trivial_cases}
Let $p$, $q$, $\varepsilon$, $\cN$, and $\cM$ be as stated in Definition~\ref{def:binary_symmetric_C}.
    If one of the following conditions hold:
    \begin{enumerate}
        \item  $\varepsilon \in [1/2,1]$,
        \item $\exists\, s\in[0,1]$ such that $\varepsilon \geq p^s q^{1-s}$,
        \item $\exists$ $\psi_{RA}$ with the dimensions of the systems $R$ and $A$ being equal such that \\ 
        $\operatorname{supp}\!\left( \cN_{A \to B} (\psi_{RA}) \right) \cap \operatorname{supp}\!\left( \cM_{A \to B} (\psi_{RA}) \right) = \emptyset$,
    \end{enumerate}
    then 
	\begin{align}
		\label{eq:binary_symmetric1_C}
		n^*(p,\cN,q,\cM, \varepsilon) = 1.
	\end{align}
	If $\cN = \cM$ and $\min\{p,q\} > \varepsilon \in [0,1/2)$, then
	\begin{align}
		\label{eq:binary_symmetric2_C}
		n^*(p,\cN,q,\cM, \varepsilon) = +\infty.
	\end{align}   
\end{proposition}
\begin{proof}
    See Appendix~\ref{app:proof-remark-triv-cond_C}.
\end{proof}

Define the following quantity for $s\in [0,1]$:
\begin{equation} \label{eq:C_N_S}
    C_s(\cN \Vert \cM) \coloneqq - \inf_{\rho_{RA}} \ln \Tr\!\left[ \cN_{A \to B}(\rho_{RA}) ^s \cM_{A \to B}(\rho_{RA})^{1-s}\right].
\end{equation}
\begin{theorem}[Query Complexity: Symmetric Binary Channel Discrimination] \label{theorem:binary_symmetric_C}
	Let $p$, $q$, $\varepsilon$, $\cN$, and $\cM$ be as stated in Definition~\ref{def:binary_symmetric_C} such that the conditions in Proposition~\ref{Prop:Trivial_cases} do not hold.
	Then the following bounds hold:
	\begin{align}
		\label{eq:binary_symmetric3_C}
		\max\left\{ \frac{\ln\!\left( \frac{pq}{\varepsilon \left(1-\varepsilon\right)} \right)}{ -\ln \widehat{F}(\cN,\cM) } ,\frac{1-\frac{\varepsilon(1-\varepsilon)}{pq}}{ \left[d_{\widehat{F}}(\cN,\cM)\right]^2  } \right\} \leq n^*(p,\cN,q,\cM, \varepsilon)
		\leq \left \lceil \inf_{s\in\left[  0,1\right]  }\frac{\ln\!\left(  \frac
{p^{s}q^{1-s}}{\varepsilon}\right)  }{C_s(\cN \Vert  \cM)}\right\rceil,
	\end{align}
where $C_s(\cN \Vert \cM)$ is defined in~\eqref{eq:C_N_S}, $\widehat{F}(\cN,\cM)$ corresponds to the channel fidelity in~\eqref{eq:channel_fidelities} with ${\bf{F}} \equiv \widehat{F}$ in~\eqref{eq:intro-geo-fid-def},
and $d_{\widehat{F}}(\cN, \cM)$ corresponds to the channel divergence in~\eqref{eq:channel_divergences} with ${\bf{D}} \equiv d_{\widehat{F}}$ in~\eqref{eq:d_B_hat}.
\end{theorem}

\begin{proof}
    See Appendix~\ref{app:binary_symmetric_C}. 
\end{proof}

\begin{corollary} \label{cor:symmetric_binary}
   Let $p$, $q$, $\varepsilon$, $\cN$, and $\cM$ be as stated in Definition~\ref{def:binary_symmetric_C} such that the conditions in Proposition~\ref{Prop:Trivial_cases} do not hold.
	Then the following bounds hold:
	\begin{align}
		\frac{\ln\!\left( \frac{pq}{\varepsilon \left(1-\varepsilon\right)} \right)}{ -\ln \widehat{F}(\cN,\cM) }  \leq n^*(p,\cN,q,\cM, \varepsilon)
		\leq 
        \left\lceil \frac{2 \ln \!\left( \frac{\sqrt{p q} }{ \varepsilon } \right) }{-\ln  F_H(\cN,\cM)} \right\rceil,
	\end{align}
where $\widehat{F}(\cN,\cM)$ and $F_H(\cN, \cM)$ correspond to channel fidelities in~\eqref{eq:channel_fidelities} with ${\bf{F}} \equiv \widehat{F}$ in~\eqref{eq:intro-geo-fid-def} and ${\bf{F}} \equiv F_H$ in~\eqref{eq:Holevo_fidelity}, respectively. 
\end{corollary}

\begin{proof}
    The proof follows directly from Theorem~\ref{theorem:binary_symmetric_C} together with the substitution $s=1/2$ in the upper bound therein.
\end{proof}

For classical channels, where the input and output systems of a channel are classical, the query complexity is tight, and it is a function of the classical channel fidelities.
\begin{corollary}[Query Complexity: Classical Channels] \label{Cor:classical_channels}
Let $p$, $q$, $\varepsilon$, $\cN$, and $\cM$ be as stated in Definition~\ref{def:binary_symmetric_C} such that the conditions in Proposition~\ref{Prop:Trivial_cases} do not hold.
   Additionally, let $\cN$ and $\cM$ be classical channels. Then we have that 
    \begin{equation}
         n^*(p,\cN,q,\cM, \varepsilon) = \Theta\!\left( \frac{\ln\!\left(\frac{1}{\varepsilon}\right) }{ -\ln F(\cN,\cM) }  \right),
    \end{equation}
where $F(\cN, \cM)$ corresponds to the channel fidelities in~\eqref{eq:channel_fidelities}.
\end{corollary}

\begin{proof}
    The proof follows by applying 
Corollary~\ref{cor:symmetric_binary} along with $1-\varepsilon \leq 1$ and identifying that, for classical channels with classical inputs, Lemma~\ref{lem:channel_divergence_classical} implies the following equality: $F(\cN,\cM) = \widehat{F}(\cN,\cM) =F_H(\cN,\cM)$.
\end{proof}

\begin{lemma}[Channel Divergences for Classical Channels] \label{lem:channel_divergence_classical}
    Let $\cN$ and $\cM$ be classical channels (classical inputs $x \in \mathcal{X}$ and classical outputs $y \in \mathcal{Y}$), which correspond to conditional probability distributions $P^\cN_{Y|X}$ and $P^\cM_{Y|X}$, respectively. Then the channel fidelities defined in~\eqref{eq:channel_fidelities} reduce as follows:
    \begin{align}
        F(\cN,\cM) & = \inf_{x \in \mathcal{X}} F\!\left( \cN(|x \rangle\!\langle x|),  \cM(|x\rangle\!\langle x|) \right) \\ 
        &= \inf_{x \in \mathcal{X}}  \left( \sum_{y \in \mathcal{Y}} \sqrt{ P^\cN_{Y|X}(y|x) P^\cM_{Y|X}(y|x) } \right)^2.
    \end{align}
Furthermore, in this case,
\begin{equation}
F(\cN,\cM)=F_H(\cN,\cM)=\widehat{F}(\cN,\cM).    
\end{equation}
\end{lemma}

\begin{proof}
    Define the classical channels $\cN$ and $\cM$ as follows: for $\mathcal{R} \in \{ \mathcal{N}, \mathcal{M} \}$
    \begin{align}
        \mathcal{R}(Z_A) \coloneqq \sum_{x,y}  P^{\mathcal{R}}_{{Y|X}}{(y|x)} \ \langle x| Z_A |x \rangle  \ |y\rangle\!\langle y|.
    \end{align}
Then 
\begin{equation}
    \mathcal{R}( \rho_{RA})= \sum_{x,y} P^{\mathcal{R}}_{{Y|X}}{(y|x)} \ \langle x|_A \rho_{RA} |x \rangle_A  \otimes  |y\rangle\!\langle y|.
\end{equation}
With the above structure, for the optimization in~\eqref{eq:channel_fidelities}, it is sufficient to restrict the optimization to only classical-quantum states of the form 
\begin{equation}
    \rho_{RA} = \sum_{x} p(x) \rho^x \otimes |x\rangle\!\langle x|.
\end{equation}
With that, consider that
\begin{align}
    \sqrt{F\!\left( \cN(\rho_{RA}), \cM(\rho_{RA}) \right)} & = \sqrt{F\!\left(\sum_{x} p(x) \rho^x \otimes \cN(|x\rangle\!\langle x|), \sum_{x} p(x) \rho^x \otimes \cM(|x\rangle\!\langle x|)\right)} \\ 
    & \geq \sum_{x} p(x) \sqrt{F\!\left( \rho^x \otimes \cN(|x\rangle\!\langle x|), \rho^x \otimes \cM(|x\rangle\!\langle x|)\right)} \\ 
   & \geq \inf_{x \in \mathcal{X}}\sqrt{F\!\left( \rho^x \otimes \cN(|x\rangle\!\langle x|), \rho^x \otimes \cM(|x\rangle\!\langle x|)\right)} \\ 
   &= \inf_{x \in \mathcal{X}} \sqrt{F\!\left(  \cN(|x\rangle\!\langle x|), \cM(|x\rangle\!\langle x|)\right)}
\end{align}
where the first inequality follows from the joint concavity of root fidelity, and the last equality by the data-processing inequality for fidelity (using the partial trace channel and the channel appending the state $\rho^x$).

This then leads to the following equality for classical channels $\mathcal{N}$ and $\mathcal{M}$:
\begin{equation}
    F(\cN, \cM) = \inf_{x \in \mathcal{X}} {F\!\left(  \cN(|x\rangle\!\langle x|), \cM(|x\rangle\!\langle x|)\right)}.
\end{equation}
Using the fact that the square root is a monotonic function, we arrive at the second equality by recalling the definition of fidelity in~\eqref{eq:fidelity}.

We conclude the proof by noting that all operators in the definition of fidelities are equal for commuting states, which arise with classical channels, together with the joint concavity of $\sqrt{F_H}$ and $\sqrt{\widehat{F}}$, similar to $\sqrt{F}$.
\end{proof}

Next, we also see that it is possible to obtain a tight characterization of the query complexity of classical--quantum channels. To this end, we use amortized channel divergences defined in~\cite{wilde2020amortized} to arrive at the said characterization. Note that~\cite{Li_25_query_channel} has used amortized channel divergences for characterizing query complexity and provided partial results on the special case of classical--quantum channels.  
Before stating the result, we define several quantities that are useful in this context. 
Define the amortized channel divergence in terms of a generalized divergence as follows: for channels $\cN_{A \to B}$ and $\cM_{A \to B}$ 
\begin{equation} \label{eq:amortized_channel_general}
    {\bf{D}}^\cA(\cN \Vert \cM) \coloneqq \sup_{\rho_{RA}, \sigma_{RA}} {\bf{D}}\!\left( \cN_{A \to B} (\rho_{RA}) \Vert \cM_{A \to B} (\sigma_{RA})\right) - {\bf{D}}(\rho_{RA} \Vert \sigma_{RA}).
\end{equation}
A classical--quantum channel $\cN$ can be defined in terms of $( \omega_x)_{x \in \mathcal{X}}$, where $\omega_x$ is a quantum state for all $x \in \mathcal{X}$:
\begin{equation}
    \cN(\cdot)= \sum_{x \in \mathcal{X}} \langle x| \cdot | x \rangle \omega_x.
\end{equation}

\begin{corollary}[Query Complexity: Classical--Quantum Channels] \label{Cor:query_CQ_1}
Let $p$, $q$, $\varepsilon$, $\cN$, and $\cM$ be as stated in Definition~\ref{def:binary_symmetric_C} such that the conditions in Proposition~\ref{Prop:Trivial_cases} do not hold.
   Additionally, let $\cN$ and $\cM$ be classical--quantum channels such that $\cN$ corresponds to $\left( \omega_x\right)_{x \in \mathcal{X}}$, and $\cM$ corresponds to $\left( \nu_x\right)_{x \in \mathcal{X}}$. Then we have that 
    \begin{equation}
        \frac{\ln\!\left( \frac{pq}{\varepsilon \left(1-\varepsilon\right)} \right)}{ -\ln  \min_{x \in \mathcal{X}} {F}_H(\omega_x,\nu_x) }  \leq n^*(p,\cN,q,\cM, \varepsilon)
		\leq 
        \left\lceil \frac{2 \ln \!\left( \frac{\sqrt{p q} }{ \varepsilon } \right) }{-\ln  \min_{x \in \mathcal{X}} {F}_H(\omega_x,\nu_x) } \right\rceil ,
    \end{equation}
where $F_H$ is defined in~\eqref{eq:Holevo_fidelity}.
Consequently, for constant $p$ and $q$, we have that 
\begin{equation}
    n^*(p,\cN,q,\cM, \varepsilon) = \Theta\! \left( \frac{ \ln \!\left( \frac{1 }{ \varepsilon } \right) }{-\ln  \min_{x \in \mathcal{X}} {F}_H(\omega_x,\nu_x) } \right).
\end{equation}
\end{corollary}

\begin{proof}
For the upper bound, recall from~\cite[Corollary~II.8]{cheng2024invitation} and Theorem~\ref{theorem:binary_symmetric_improved_lower_bnd}, and together with that, the sample complexity of state discrimination as defined in~\eqref{eq:SC_states_binary} satisfies the following inequalities: 
\begin{equation} \label{eq:SC_from_previous}
 \frac{\ln\!\left( \frac{pq}{\varepsilon\left(1-\varepsilon\right)} \right)}{ -\ln {F}_H(\rho, \sigma) }  \leq   n^*(p, \rho,q ,\sigma, \varepsilon) \leq \left\lceil \frac{2 \ln \!\left( \frac{\sqrt{p q} }{ \varepsilon } \right) }{-\ln  {F}_H(\rho,\sigma) } \right\rceil,
\end{equation}
where the upper bound is achieved by product measurements (independent measurements on each copy of the state) and classical post-processing.

Let us consider the following protocol, in which we use a state discrimination protocol for the discrimination of classical--quantum channels. This is done by choosing the input to the unknown channel to be $x \in \mathcal{X}$ for every query to it and then using a state discrimination protocol to discriminate the states $\omega_x^{\otimes n}$ and $\nu_x^{\otimes n}$. 
For this setup, choosing the state discrimination protocol to be the one that achieves the upper bound of~\eqref{eq:SC_from_previous} leads to
\begin{equation}
    n^*(p,\cN,q,\cM, \varepsilon) \leq \left\lceil \frac{2 \ln \!\left( \frac{\sqrt{p q} }{ \varepsilon } \right) }{-\ln  {F}_H(\omega_x,\nu_x) } \right\rceil. 
\end{equation}
By optimizing over all $x \in \mathcal{X}$, we arrive at the desired upper bound. 

\medskip Before going over the proof of the lower bound, let us recall the following properties:
By choosing~${\bf{D}}$ in~\eqref{eq:amortized_channel_general} as $-\ln F_H$ and considering its definition, we have that for all $\rho_{RA}, \sigma_{RA}$
\begin{equation}
    -\ln F_H^\cA(\cN \Vert \cM) \geq {-\ln F_H}\!\left( \cN_{A \to B} (\rho_{RA}) \Vert \cM_{A \to B} (\sigma_{RA})\right) + {\ln F_H}(\rho_{RA} \Vert \sigma_{RA}).
\end{equation}
This leads to the chain rule 
\begin{equation}\label{eq:chain_rule_F_H_amortized}
  { F_H}\!\left( \cN_{A \to B} (\rho_{RA}) \Vert \cM_{A \to B} (\sigma_{RA})\right)    \geq  F_H^\cA(\cN \Vert \cM) F_H(\rho_{RA} \Vert \sigma_{RA}).
\end{equation}

To prove the lower bound, recall~\eqref{eq:def:sample_complexity_symmetric2_C} and use~\eqref{eq:start_for other proofs}, so that we arrive at
\begin{align}
    \frac{\varepsilon\left(1-\varepsilon\right)}{pq} 
    &\geq \inf_{\{ \cA\},\rho_{R_1A_1}}  F\!\left( \rho_{R_n B_n}, \tau_{R_n B_n}\right) \\ 
    &\geq \inf_{\{ \cA\},\rho_{R_1A_1}}  F_H\!\left( \rho_{R_n B_n}, \tau_{R_n B_n}\right),
\end{align}
where the last inequality holds from $F_H \leq F$.

Then, by applying similar proof arguments as in the proof of the first lower bound in Theorem~\ref{theorem:binary_symmetric_C} (i.e., steps \eqref{eq:similar_start}--\eqref{eq:bound_on-error}) by replacing $\widehat{F}$ with $F_H$ together with~\eqref{eq:chain_rule_F_H_amortized} and the data-processing property of Holevo fidelity, we arrive at 
\begin{equation}
    n^*(p,\cN,q, \cM, \varepsilon)
    \geq \frac{\ln\!\left( \frac{pq}{\varepsilon\left(1-\varepsilon\right)} \right)}{ -\ln {F}_H^\cA( \cN, \cM) }. 
\end{equation}
Finally, by applying~\cite[Lemma~26]{wilde2020amortized}, which is valid for classical--quantum channels and for the Petz-R\'enyi divergence ($\alpha \in (0,2)$), with the choice $\alpha=1/2$, we have
\begin{equation}
   {F}_H^\cA( \cN, \cM) = \min_{x \in \mathcal{X}} F_H(\omega_x, \nu_x),
\end{equation}
completing the proof of the lower bound. 
\end{proof}

\begin{remark}[Optimal Strategies for Classical--Quantum Channels]
    We remark that for classical--quantum channels, from Corollary~\ref{Cor:query_CQ_1} and for fixed $p$ and $q$, the optimal asymptotic sample complexity is achieved by a product strategy (picking the best possible input and then applying a tensor-power strategy), and there is no advantage from adaptive strategies. Also, we see that the reference system $R$ is not required to achieve the optimal sample complexity asymptotically.
\end{remark}

\begin{remark}[Special Classes of Channels]
    For some special classes of channels, including environment-seizable channels~\cite[Definition~36]{wilde2020amortized}, the query complexity bounds boil down to quantum hypothesis testing of states with the environment states as the states to be distinguished. In that case, the results derived in~\cite{cheng2024invitation} can be directly used to derive bounds on the query complexity. Furthermore, for the task of discriminating bosonic dephasing channels, it is possible to provide analytical expressions for the query complexity, by utilizing the findings in~\cite{HuangBosonciDephasing24}.
\end{remark}

\subsubsection{Alternative Characterization of Query Complexity}

In this section, we provide an alternative characterization of the query complexity of channel discrimination, by applying Theorem~\ref{thm:alterantive_SC_binary}. Towards this goal, define the following quantities for $s \in (0,1)$:
\begin{align}
    \widehat{Q}_s(\rho \Vert \sigma) & \coloneqq \Tr\!\left[ \sigma \left(\sigma^{-1/2} \rho \sigma^{-1/2} \right)^s \right], \\
    \widehat{Q}_s(\cN \Vert \cM) & \coloneqq \inf_{\rho_{RA}} \widehat{Q}_s\! \left( \cN_{A \to B} (\rho_{RA}) \Vert \cM_{A \to B} (\rho_{RA}) \right) \label{eq:Q_hat_s_channel}, \\
    {Q}_s(\cN \Vert \cM) & \coloneqq \inf_{\rho_{RA}} Q_s\! \left( \cN_{A \to B} (\rho_{RA}) \Vert \cM_{A \to B} (\rho_{RA}) \right) \label{eq:Q_s_channel}, 
\end{align}
where $Q_s(\cdot \Vert \cdot)$ is defined in~\eqref{eq:Q_s}. 
Also note that, by~\eqref{eq:C_N_S},
\begin{equation}
    C_s(\cN \Vert \cM) = -\ln Q_s(\cN \Vert \cM).
\end{equation}
\begin{theorem} \label{prop:alternative_query_C}
    Let $p \in (0,1/2]$, $q \equiv 1-p$, $\varepsilon \in (0, p/64)$, and let $\cN$ and $\cM$ be as stated in Definition~\ref{def:binary_symmetric_C}. We have the following: 
    \begin{equation}
        \left \lceil \frac{1}{2} \lambda_* \frac{ \ln\!\left( \frac{p}{\varepsilon} \right)}{-\ln \widehat{Q}_{\lambda^*}(\cN \Vert \cM)} \right \rceil   \leq  n^*(p,\cN,q,\cM,\varepsilon) \leq \left \lceil 2 \lambda_* \frac{ \ln\!\left( \frac{p}{\varepsilon} \right)}{-\ln Q_{\lambda^*}(\cN \Vert \cM)} \right \rceil, 
    \end{equation}
where $Q_{\lambda_*}(\cdot \Vert \cdot)$, $\widehat{Q}_{\lambda_*}(\cdot \Vert \cdot)$, and $\lambda_*$ are defined in~\eqref{eq:Q_s_channel},~\eqref{eq:Q_hat_s_channel}, and~\eqref{eq:lambda_*}, respectively. 
\end{theorem}
\begin{proof}
Recall from~\eqref{eq:Helstrom2_C} that
\begin{equation}
    p_e(p, \cN, q, \cM, n) =\inf_{ \cA,\rho_{R_1A_1}} \frac12 \left( 1 - \left\| p \rho_{R_n B_n} - q \tau_{R_n B_n} \right\|_1 \right).
\end{equation}
Then, by choosing $n$ copies of a state $\rho_{RA}$ with $\rho_{R_n B_n}= \left(\cN_{A \to B}(\rho_{RA})\right)^{\otimes n}$ and $\tau_{R_n B_n}= \left(\cM_{A \to B}(\rho_{RA})\right)^{\otimes n}$, one can obtain an achievable protocol $\cA$ by devising a strategy to discriminate between $n$ copies of the states $\cN_{A \to B}(\rho_{RA})$ and $\cM_{A \to B}(\rho_{RA})$. By using the upper bound of Theorem~\ref{thm:alterantive_SC_binary}, if 
\begin{equation}
    m(\rho_{RA},\cN, \cM, p, \varepsilon) \coloneqq  \left \lceil 2 \lambda_* \frac{ \ln\!\left( \frac{p}{\varepsilon} \right)}{-\ln Q_{\lambda^*}\! \left(\cN(\rho_{RA}) \Vert \cM(\rho_{RA}) \right)} \right \rceil,
\end{equation}
the error probability of state discrimination for this specific scenario is at most $\varepsilon$. By optimizing over the state $\rho_{RA}$, it is possible to design a protocol such that the error criterion is met by choosing 
\begin{align}
    n &= \inf_{\rho_{RA}}  m(\rho_{RA},\cN, \cM, \varepsilon) \\
    &=  \left \lceil 2 \lambda_* \frac{ \ln\!\left( \frac{p}{\varepsilon} \right)}{-\ln  \inf_{\rho_{RA}} Q_{\lambda^*}\! \left(\cN(\rho_{RA} \Vert \cM(\rho_{RA}) \right)} \right \rceil \\
    &=  \left \lceil 2 \lambda_* \frac{ \ln\!\left( \frac{p}{\varepsilon} \right)}{-\ln Q_{\lambda^*}\! \left(\cN \Vert \cM \right)} \right \rceil, 
\end{align}
concluding the proof of the upper bound. 

\medskip
Before proving the lower bound, we state the following facts that are useful in the proof. 
By~\cite[Proposition~7.49]{KW21}, we have that for $\alpha \in (0,1) \cup (1,2]$
\begin{equation}
    D_\alpha(\rho \Vert \sigma) \leq \widehat{D}_\alpha(\rho \Vert \sigma),
\end{equation}
where $D_\alpha$ and $\widehat{D}_\alpha$ are defined in~\eqref{eq:petz-renyi-div} and~\eqref{eq:geometric-renyi-div}, respectively.
Using that, for $s\in(0,1)$, we have that 
\begin{equation}\label{eq:Q_s_hat_Q}
    -\ln Q_s(\rho \Vert \sigma) \leq -\ln \widehat{Q}_s(\rho \Vert \sigma).
\end{equation}
Also, by the chain rule for the geometric R\'enyi channel divergence (recall~\eqref{eq:chain_rule_geometric_renyi}), we have that
\begin{equation}
  -\ln   \widehat{Q}_s\!\left( \cN_{A \to B}(\rho_{RA}) \Vert \cM_{A \to B} (\sigma_{RA})\right) \leq -\ln \widehat{Q}_s(\cN \Vert \cM) - \ln \widehat{Q}_s(\rho_{RA} \Vert \sigma_{RA}),
\end{equation}
which then leads to the following multiplicative property (chain rule for $\widehat{Q}_s$):
\begin{equation} \label{eq:chain_multi_Q_s}
     \widehat{Q}_s\!\left( \cN_{A \to B}(\rho_{RA}) \Vert \cM_{A \to B} (\sigma_{RA})\right) \geq \widehat{Q}_s(\cN \Vert \cM) \  \widehat{Q}_s(\rho_{RA} \Vert \sigma_{RA}).
\end{equation}
Also, note that $Q_s$ satisfies the data-processing inequality for all $s\in(0,1)$:
\begin{equation} \label{eq:DPI_QS}
    \widehat{Q}_s(\rho \Vert \sigma) \leq \widehat{Q}_s\!\left( \cN(\rho) \Vert \cN(\sigma) \right).
\end{equation}


Note that the channel discrimination task reduces to the task of discriminating states $\rho_{R_nA_n}$ and $\sigma_{R_n A_n}$, which are generated by protocol $\cA$ and the initial state $\rho_{R_1A_1}$, and optimizing over $\{\cA, \rho_{R_1A_1}\}$ (see Section~\ref{subS:symmetric_binary}). 
Then, recalling~\eqref{eq:def:sample_complexity_symmetric2_C} and using~\eqref{eq:start_for_channel_proof_Q_s}, we have that
\begin{align}
    \varepsilon & \geq  \inf_{\{ \cA\},\rho_{R_1A_1}} \frac12 \left( 1 - \left\| p \rho_{R_n B_n} - q \tau_{R_n B_n} \right\|_1 \right)\\ & =  \inf_{\cA, \rho_{R_1 A_1}} p_e( p, \rho_{R_nA_n}, q, \sigma_{R_n A_n}, 1) \\ 
    & \geq  \inf_{\cA, \rho_{R_1 A_1}} \left(\frac{p}{\varepsilon} \right)^{\lambda_*}  \varepsilon \left[ Q_{\lambda_*}(\rho_{R_n A_n} \Vert \sigma_{R_n A_n})\right]^{2} \label{eq:start_Q_L_CQ} \\ 
    & \geq \inf_{\cA, \rho_{R_1 A_1}} \left(\frac{p}{\varepsilon} \right)^{\lambda_*}  \varepsilon \left[ \widehat{Q}_{\lambda_*}(\rho_{R_n A_n} \Vert \sigma_{R_n A_n})\right]^{2},
\end{align}
where the last inequality follows from~\eqref{eq:Q_s_hat_Q}.


Then, by applying similar proof arguments as in the proof of the first lower bound in Theorem~\ref{theorem:binary_symmetric_C} (i.e., steps \eqref{eq:similar_start}--\eqref{eq:bound_on-error}) by replacing $\widehat{F}$ with $\widehat{Q}_{\lambda_*}$ together with the chain rule in~\eqref{eq:chain_multi_Q_s} and the data-processing property in~\eqref{eq:DPI_QS}, completes the proof of the lower bound. 
\end{proof}

Using the above characterization of query complexity, we improve the query complexity of classical channel discrimination in this setting, in contrast to Corollary~\ref{Cor:classical_channels}, as presented next. 

\begin{corollary}[Query Complexity of Classical Channel Discrimination]
\label{cor:precise-classical}
 Let $p$, $q$, $\varepsilon$, $\cN$, and $\cM$ be as stated in Definition~\ref{def:binary_symmetric_C} such that $p \in (0,1/2]$, $q \equiv 1-p$, $\varepsilon \in (0, p/64)$, and $\cN$ and $\cM$ are classical-input, classical-output channels. Then, we have that
 
    \begin{equation}
        \left \lceil \frac{1}{2} \lambda_* \frac{ \ln\!\left( \frac{p}{\varepsilon} \right)}{-\ln {Q}_{\lambda^*}(\cN \Vert \cM)} \right \rceil   \leq  n^*(p,\cN,q,\cM,\varepsilon) \leq \left \lceil 2 \lambda_* \frac{ \ln\!\left( \frac{p}{\varepsilon} \right)}{-\ln Q_{\lambda^*}(\cN \Vert \cM)} \right \rceil, 
    \end{equation}
where $Q_{\lambda_*}(\cdot \Vert \cdot)$, and $\lambda_*$ are defined in~\eqref{eq:Q_s_channel} and~\eqref{eq:lambda_*}, respectively.  

Consequently, we have that 
\begin{equation}
     n^*(p,\cN,q,\cM,\varepsilon) = \Theta\! \left(\lambda_* \frac{ \ln\!\left( \frac{p}{\varepsilon} \right)}{-\ln {Q}_{\lambda^*}(\cN \Vert \cM)}  \right).
\end{equation}
\end{corollary}

\begin{proof}
    The proof follows by applying Proposition~\ref{prop:alternative_query_C} together with the equality 
    $Q_s(\cN \Vert \cM)= \widehat{Q}_s(\cN \Vert \cM)$ for classical channels. The proof of the latter fact follows similarly to the proof of Lemma~\ref{lem:channel_divergence_classical} along with the quasi-convexity of Petz-R\'enyi divergence and geometric R\'enyi divergence.
\end{proof}

Furthermore, we extend our characterization of query complexity to classical--quantum channels in the next corollary. 

\begin{corollary}[Query Complexity of Classical--Quantum Channel Discrimination] \label{Cor:Q_C_CQ_2}
 Let $p$, $q$, $\varepsilon$, $\cN$, and $\cM$ be as stated in Definition~\ref{def:binary_symmetric_C} such that $p \in (0,1/2]$, $q \equiv 1-p$, $\varepsilon \in (0, p/64)$, and $\cN$ and $\cM$ are classical--quantum channels corresponding to $\left( \omega_x\right)_{x \in \mathcal{X}}$ and $\left( \nu_x\right)_{x \in \mathcal{X}}$, respectively. Then, we have that
    \begin{equation}
        \left \lceil \frac{1}{2} \lambda_* \frac{ \ln\!\left( \frac{p}{\varepsilon} \right)}{-\ln  \min_{x \in \mathcal{X}} {Q}_{\lambda^*}(\omega_x \Vert \nu_x)} \right \rceil   \leq  n^*(p,\cN,q,\cM,\varepsilon) \leq \left \lceil 2 \lambda_* \frac{ \ln\!\left( \frac{p}{\varepsilon} \right)}{-\ln \min_{x \in \mathcal{X}} {Q}_{\lambda^*}(\omega_x \Vert \nu_x)}\right \rceil, 
    \end{equation}
where $Q_{\lambda_*}(\cdot \Vert \cdot)$, and $\lambda_*$ are defined in~\eqref{eq:Q_s_channel} and~\eqref{eq:lambda_*}, respectively.  

Consequently, we have that 
\begin{equation}
     n^*(p,\cN,q,\cM,\varepsilon) = \Theta\! \left(\lambda_* \frac{ \ln\!\left( \frac{p}{\varepsilon} \right)}{-\ln \min_{x \in \mathcal{X}} {Q}_{\lambda^*}(\omega_x \Vert \nu_x)}  \right).
\end{equation}
\end{corollary}
\begin{proof}
The proof follows similarly to the proof of Corollary~\ref{Cor:query_CQ_1} by making use of Theorem~\ref{thm:alterantive_SC_binary} derived in this work, instead of~\eqref{eq:SC_from_previous}, which was proven in \cite[Corollary~II.8]{cheng2024invitation},  together with the amortized channel divergence with the instantiation $ {\bf{D}} \equiv -\ln Q_{\lambda_*}$, \eqref{eq:start_Q_L_CQ} in the proof of Theorem~\ref{prop:alternative_query_C}, and \cite[Lemma~26]{wilde2020amortized} for $\alpha= \lambda^*$.
\end{proof}


\begin{remark}
     Theorem~\ref{theorem:binary_symmetric_C} and Theorem~\ref{prop:alternative_query_C} provide different characterizations of lower and upper bounds on the query complexity for various ranges of error parameters $\varepsilon$ and priors $p$.  The main novelty of Theorem~\ref{prop:alternative_query_C}  over Theorem~\ref{theorem:binary_symmetric_C} is that its reduction to classical channels and classical--quantum channels improves Corollary~\ref{Cor:classical_channels} and Corollary~\ref{Cor:query_CQ_1}, which are based on Theorem~\ref{theorem:binary_symmetric_C}, for the regime $\varepsilon < p/64$ and for $p \leq 1/2$. In fact, in this regime, Theorem~\ref{prop:alternative_query_C} leads to Corollary~\ref{cor:precise-classical} and Corollary~\ref{Cor:Q_C_CQ_2}, which show that the lower bounds and upper bounds differ only by a constant factor of four. Also, in contrast to the previous bounds, we obtain a tight characterization of query complexity with respect to all variables of interest, including the \textit{priors}. 
\end{remark}

\begin{remark}
We remark from Corollary~\ref{Cor:Q_C_CQ_2} that for sufficiently small error $\varepsilon$ (specifically, $\varepsilon \in (0, p/64)$), the task of discriminating between two classical--quantum channels can be optimally addressed using a product strategy; namely, selecting the best possible input and applying a tensor-power strategy. In this regime, the query complexity increases by at most a factor of four compared to the optimal bound, making this strategy effectively (non-asymptotically) optimal. In contrast, Corollary~\ref{Cor:query_CQ_1} does not offer such guarantees, as the gap between its upper and lower bounds varies with the priors. The results in Corollary~\ref{Cor:query_CQ_1} treat priors as fixed constants, and thus only yield asymptotic optimality. We use the term asymptotic to refer to scenarios where constants are ignored and only the scaling behavior with respect to parameters is considered. In contrast, non-asymptotic refers to results that preserve constant factors and offer concrete optimality guarantees.
   
   Moreover, Corollary~\ref{Cor:Q_C_CQ_2} indicates that access to the reference system is not required to achieve optimal query complexity up to a constant factor of four.
\end{remark}

\subsection{Asymmetric Channel Discrimination}
In this section, we characterize the query complexity of asymmetric channel discrimination of two channels defined in Definition~\ref{def:binary_asymmetric_C}. 
\begin{theorem}[Query Complexity: Asymmetric Channel Discrimination]
\label{thm:binary_asymmetric_C} Fix $\varepsilon,\delta \in (0,1)$, and let $\cN$ and $\cM$ be quantum channels. 
Then the following  bounds hold for the query complexity $n^{\ast}(\cN,\cM,\varepsilon,\delta)$ of asymmetric binary channel discrimination defined in Definition~\ref{def:binary_asymmetric_C}:
\begin{multline}
\max\left\{  \sup_{ \alpha \in (1,2] } \left(  \frac{\ln\!\left(  \frac{\left(  1-\varepsilon\right)
^{\alpha^{\prime}}}{\delta}\right)  }{\widehat{D}_{\alpha}(\cN\Vert\cM
)}\right)  ,\ \sup_{ \alpha \in (1,2] } \left(  \frac{\ln\!\left(  \frac{\left(
1-\delta\right)  ^{\alpha^{\prime}}}{\varepsilon}\right)  }{\widehat
{D}_{\alpha}(\cM \Vert\cN)}\right)  \right\}  \leq n^{\ast}(\cN,\cM,\varepsilon,\delta) \\ \leq 
\min\left\{  \left\lceil
\inf_{\alpha\in\left(  0,1\right)  }\left(  \frac{\ln\!\left(  \frac
{\varepsilon^{\alpha^{\prime}}}{\delta}\right)  }{D_{\alpha}(\cN\Vert\cM
)}\right)  \right\rceil ,\left\lceil \inf_{\alpha\in\left(  0,1\right)
}\left(  \frac{\ln\!\left(  \frac{\delta^{\alpha^{\prime}}}{\varepsilon
}\right)  }{D_{\alpha}(\cM \Vert\cN)}\right)  \right\rceil \right\}  .
\label{eq:binary_asymmetric_samp_comp_C}
\end{multline}
where $\alpha^{\prime}\coloneqq\frac{\alpha}{\alpha-1}$, and the channel divergence of $\widehat{D}_\alpha$ and $D_\alpha$ are defined as in~\eqref{eq:channel_divergences} by choosing the corresponding divergence in place of the generalized divergence ${\bf{D}}$.
\end{theorem}

\begin{proof}
See Appendix~\ref{app:proof_samp_comp_binary_asymmetric_C}.
\end{proof}

For the state discrimination problem, it has been shown that the sample complexity of binary hypothesis testing in the symmetric and the asymmetric case are related as a function of the prior probabilities and fixed error. To this end, this has been shown for the classical setting in~\cite{sample_complexity_classical2024} and, subsequently, for the quantum setting in~\cite{cheng_privateQHT24}.
We next show that the query complexities of symmetric and asymmetric channel discrimination are also related in a similar fashion, and one can then also use the query complexity bounds derived in Proposition~\ref{Prop:Trivial_cases}, Theorem~\ref{theorem:binary_symmetric_C}, and Theorem~\ref{prop:alternative_query_C}  to derive bounds on the query complexity of the asymmetric setting.

\begin{proposition} [Relationship between Query Complexity of Symmetric and Asymmetric Channel Discrimination] \label{prop:Rel_AS_S}
    Let $\varepsilon, \delta \in (0,1)$, and let $\cN$, $\cM$ be quantum channels. Then, we have that 
    \begin{equation}
        n^*_{\mathsf{S}}\!\left( \frac{\delta}{\varepsilon +\delta} , \cN, \frac{\varepsilon}{\varepsilon +\delta}, \cM, \frac{2\varepsilon\delta}{\varepsilon +\delta} \right) \leq n^*_{\mathsf{AS}}\!\left( \cN, \cM, \varepsilon, \delta \right) \leq n^*_{\mathsf{S}}\!\left( \frac{\delta}{\varepsilon +\delta} , \cN, \frac{\varepsilon}{\varepsilon +\delta}, \cM, \frac{\varepsilon\delta}{\varepsilon +\delta} \right), \label{eq:first_rel_S_AS}
    \end{equation}
where $n^*_{\mathsf{S}} \equiv n^*$ as defined in Definition~\ref{def:binary_symmetric_C} and $n^*_{\mathsf{AS}} \equiv n^*$ as defined in Definition~\ref{def:binary_asymmetric_C}.

Furthermore, we also have that
for $p, \varepsilon \in (0,1)$:
 \begin{equation} \label{eq:second_rel_S_AS}
     n^*_{\mathsf{AS}}\!\left( \cN, \cM, \frac{\varepsilon}{p} ,\frac{\varepsilon}{1-p}  \right)   \leq     n^*_{\mathsf{S}}\!\left( p, \cN, 1-p, \cM, \varepsilon \right)  \leq n^*_{\mathsf{AS}}\!\left( \cN, \cM, \frac{\varepsilon}{2p} ,\frac{\varepsilon}{2(1-p)}  \right). 
    \end{equation}
\end{proposition}
\begin{proof}
To prove the first inequality in~\eqref{eq:first_rel_S_AS}, consider a protocol characterized by $\{Q, \cA\}$ and a state $\rho_{R_1A_1}$ such that $\Tr\!\left[ (I_{R_n B_n} -Q_{R_n B_n} )\rho_{R_n B_n}\right] \leq \varepsilon $ and $  \Tr\!\left[ Q_{R_n B_n} \tau_{R_n B_n}\right]  \leq \delta$.
Then, we see that 
\begin{equation}
    \frac{\delta} {\varepsilon + \delta} \Tr\!\left[ (I_{R_n B_n} -Q_{R_n B_n} )\rho_{R_n B_n}\right] +  \frac{\varepsilon} {\varepsilon + \delta} \Tr\!\left[ Q_{R_n B_n} \tau_{R_n B_n}\right]  \leq \frac{2 \varepsilon \delta}{\varepsilon + \delta}.
\end{equation}
This shows that this chosen protocol with fixed $n$ channel uses of the unknown channel $\cN$ or $\cM$ satisfies the following inequality by Definition~\ref{def:binary_symmetric_C}:
\begin{equation}
      n^*_{\mathsf{S}}\!\left( \frac{\delta}{\varepsilon +\delta} , \cN, \frac{\varepsilon}{\varepsilon +\delta}, \cM, \frac{2\varepsilon\delta}{\varepsilon +\delta} \right) \leq n.
\end{equation}
With that optimizing over all protocols $\{Q,\cA\}$ such that $\Tr\!\left[ (I_{R_n B_n} -Q_{R_n B_n} )\rho_{R_n B_n}\right] \leq \varepsilon $ and $  \Tr\!\left[ Q_{R_n B_n} \tau_{R_n B_n}\right]  \leq \delta$, we obtain the desired inequality. 

For the second inequality in~\eqref{eq:first_rel_S_AS}, consider a protocol $\{Q,\cA\}$ such that 
\begin{equation} \label{eq:demanded_error}
    \frac{\delta} {\varepsilon + \delta} \Tr\!\left[ (I_{R_n B_n} -Q_{R_n B_n} )\rho_{R_n B_n}\right] +  \frac{\varepsilon} {\varepsilon + \delta} \Tr\!\left[ Q_{R_n B_n} \tau_{R_n B_n}\right]  \leq \frac{ \varepsilon \delta}{\varepsilon + \delta}. 
\end{equation}
This then leads to the following inequalities:
\begin{equation}
    \Tr\!\left[ (I_{R_n B_n} -Q_{R_n B_n} )\rho_{R_n B_n}\right] \leq \varepsilon, \quad \quad \Tr\!\left[ Q_{R_n B_n} \tau_{R_n B_n}\right] \leq \delta.
\end{equation}
Furthermore, this shows that the chosen protocol $\{Q,\cA\}$ with $n$ uses of the unknown channel satisfies the following inequality by Definition~\ref{def:binary_asymmetric_C}:
\begin{equation}
    n^*_{\mathsf{AS}}\!\left( \cN, \cM, \varepsilon, \delta \right) \leq n.
\end{equation}
Then, by optimizing over all $\{Q,\cA\}$ such that~\eqref{eq:demanded_error} holds, we complete the proof of the desired inequality. 

The proof of~\eqref{eq:second_rel_S_AS} follows by using similar proof arguments as used in the proof of~\eqref{eq:first_rel_S_AS}.
\end{proof}

With the use of Proposition~\ref{prop:Rel_AS_S}, now one can obtain query complexity bounds in the asymmetric setting by applying Theorem~\ref{theorem:binary_symmetric_C}, Theorem~\ref{prop:alternative_query_C}, and Corollary~\ref {cor:symmetric_binary}. Using the simplest of those, we next obtain bounds on the query complexity of asymmetric channel discrimination. 
\begin{corollary}
    Let $\varepsilon, \delta \in (0,1)$ along with the following additional constraints: $\varepsilon \delta /(\varepsilon +\delta) \in (0,1/4)$ and for all $s \in (0,1)$, $\varepsilon^s \delta ^{1-s} < 1/2$. Also, let $\cN$ and $\cM$ be quantum channels such that for all $\psi_{RA}$, with the dimensions of the systems $R$ and $A$ being equal,  satisfy $\operatorname{supp}\!\left( \cN_{A \to B} (\psi_{RA}) \right) \cap \operatorname{supp}\!\left( \cM_{A \to B} (\psi_{RA}) \right) \neq \emptyset$. Then, we have the following bounds on the query complexity of asymmetric channel discrimination, as defined in Definition~\ref{def:binary_asymmetric_C}:
    \begin{equation}
    \frac{ \ln \!\left( \frac{1 }{ 2(\varepsilon+\delta) } \right) }{-\ln  \widehat{F}(\cN,\cM)}   \leq   n^*(\cN, \cM, \varepsilon, \delta) \leq  \left\lceil \frac{ \ln \!\left( \frac{1 }{ \varepsilon  \delta} \right) }{-\ln  F_H(\cN,\cM)} \right\rceil ,
    \end{equation}
where $\widehat{F}(\cN,\cM)$ and $F_H(\cN, \cM)$ correspond to channel fidelities in~\eqref{eq:channel_fidelities} with ${\bf{F}} \equiv \widehat{F}$ in~\eqref{eq:intro-geo-fid-def} and ${\bf{F}} \equiv F_H$ in~\eqref{eq:Holevo_fidelity}, respectively. 
\end{corollary}
\begin{proof}
    The proof for the general setting follows by utilizing Proposition~\ref{prop:Rel_AS_S} together with Proposition~\ref{Prop:Trivial_cases} and Corollary~\ref{cor:symmetric_binary}. 
%
\end{proof}

\subsection{Multiple Channel Discrimination}

In previous sections, our focus was on discriminating two quantum channels, and here we extend our results to the setting of multiple channel discrimination. 
\begin{theorem}[Query Complexity: $M$-ary Channel Discrimination] \label{theorem:sc_M-ary_C}
	Let $n^*(\mathcal{S},\varepsilon)$ be as stated in Definition~\ref{def:M-ary_C}.
	Then,
	\begin{align} \label{eq:sc_M-ary_C}
		\max_{m\neq \bar{m}}  \frac{ \ln\!\left( \frac{p_m p_{\bar m}}{ (p_m + p_{\bar m})\varepsilon } \right) }{ -\ln \widehat{F}(\cN^m, \cN^{\bar m }) }
		\leq 
		n^*(\mathcal{S},\varepsilon) \leq 
		\left\lceil  \max_{m\neq \bar{m}}  \frac{2\ln\!\left( \frac{ M(M-1) \sqrt{p_m} \sqrt{p_{\bar{m}}}  }{ 2\varepsilon } \right) }{ - \ln F\!\left(\cN^{m},\cN^{\bar{m}} \right) } \right\rceil ,
	\end{align}
where $\widehat{F}(\cN,\cM)$ and $F_H(\cN, \cM)$ for two channels $\cN$ and $\cM$ correspond to channel fidelities in~\eqref{eq:channel_fidelities} with ${\bf{F}} \equiv \widehat{F}$ in~\eqref{eq:intro-geo-fid-def} and ${\bf{F}} \equiv F_H$ in~\eqref{eq:Holevo_fidelity}, respectively. 
\end{theorem}

\begin{proof}
    See Appendix~\ref{app:proof-sc_M-ary_C}.
\end{proof}

\section{Conclusion}

In this work, we defined the query complexity of quantum channel discrimination in the following settings: binary symmetric channel discrimination, binary asymmetric channel discrimination, and multiple symmetric channel discrimination. Then, we obtained upper and lower bounds on the query complexity of these three settings. Notably, we obtained a tight characterization of the query complexity of discriminating  two classical channels and two classical--quantum channels.

It is an interesting future work to obtain tighter characterizations of query complexity for larger classes of quantum channels. Furthermore, similar to the applications of sample complexity of quantum hypothesis testing in quantum simulation, quantum learning, and quantum algorithms as explained in~\cite{cheng2024invitation}, future work includes exploring how the query complexity results derived in this work can be applied to scenarios where discrimination of channels are applicable, in contrast to quantum states. 

\begin{acknowledgments}
We thank Amira Abbas, Marco Cerezo, Ludovico Lami, and Vishal Singh for helpful discussions. We are especially grateful to Varun Jog for pointing us to \cite{kazemi2025sample}. We also thank Hyukjoon Kwon, Byeongseon Go,  and Siheon Park for notifying us of the improved lower bound given in Theorem~\ref{theorem:binary_symmetric_improved_lower_bnd}.

TN and MMW acknowledge support from NSF Grant No.~2329662.
TN acknowledges support from the Dieter Schwarz Exchange Programme on Quantum Communication and Security at the Centre for Quantum Technologies (CQT), National University of Singapore, Singapore, during her research visit to Marco Tomamichel's group at CQT.
\end{acknowledgments}

\bibliography{Reference}

\appendix

\section{Proof of Theorem~\ref{theorem:binary_symmetric_improved_lower_bnd}}

Suppose that
\begin{equation}
\frac{1}{2}\left(1-\left\Vert p\rho-q\sigma\right\Vert _{1}\right)\leq\varepsilon.
\end{equation}
Then it follows that
\begin{align}
\frac{1}{2}\left(1-\left\Vert p\rho-q\sigma\right\Vert _{1}\right) & \leq\varepsilon\\
\implies\qquad1-2\varepsilon & \leq\left\Vert p\rho-q\sigma\right\Vert _{1}\\
 & \leq\sqrt{1-4pqF(\rho,\sigma)}\\
\implies\qquad4pqF(\rho,\sigma) & \leq1-\left(1-2\varepsilon\right)^{2}\\
 & =4\varepsilon\left(1-\varepsilon\right)\\
\implies\qquad F(\rho,\sigma) & \leq\frac{\varepsilon\left(1-\varepsilon\right)}{pq}. \label{eq:take-off-step-channels}
\end{align}
The third inequality follows from  \cite[Theorem~5]{Aud14} with $A=p\rho$
and $B=q\sigma$.

Now applying the assumption that 
\begin{equation}
\frac{1}{2}\left(1-\left\Vert p\rho^{\otimes n}-q\sigma^{\otimes n}\right\Vert _{1}\right)\leq\varepsilon,
\end{equation}
plugging in tensor-power states $\rho^{\otimes n}$ and $\sigma^{\otimes n}$,
and applying the multiplicativity of fidelity, we conclude that
\begin{align}
F(\rho,\sigma)^{n} & =F(\rho^{\otimes n},\sigma^{\otimes n})\\
 & \leq\frac{\varepsilon\left(1-\varepsilon\right)}{pq},
\end{align}
which in turn implies that
\begin{align}
n\ln F(\rho,\sigma) & \leq\ln\!\left(\frac{\varepsilon\left(1-\varepsilon\right)}{pq}\right)\\
\implies\qquad n & \geq\frac{\ln\!\left(\frac{pq}{\varepsilon\left(1-\varepsilon\right)}\right)}{-\ln F(\rho,\sigma)}.
\end{align}
This is the claimed lower bound.

\label{app:proof-improved-samp-comp-lower-bnd}

\section{Proof of Proposition~\ref{Prop:Trivial_cases}}

\label{app:proof-remark-triv-cond_C}
We first prove~\eqref{eq:binary_symmetric1_C}. 

\underline{First condition:}
  Assume that $\varepsilon \in [1/2,1]$. 
	The trivial strategy of guessing the channel uniformly at random achieves an error probability of $1/2$ (i.e., with $n=1$, $ Q = I/2$ in~\eqref{eq:Helstrom1_C}). 
	Thus, if $\varepsilon \in [1/2,1]$, then this trivial strategy accomplishes the task with just a single access to the channel. 

	\underline{Second condition:}
	Now assume that there exists $s\in[0,1]$ such that $\varepsilon \geq p^s q^{1-s}$. 
	With a single use of the channel, we can achieve the error probability $\frac12 \left( 1 - \left\| p \rho_{R_1 B_1} - q \tau_{R_1B_1} \right\|_1 \right)$ for the initial state $\rho_{R_1A_1}$, and by invoking \cite[Theorem~1]{ACM+07}
 with $A = p\rho_{R_1 B_1}$ and $B = q\tau_{R_1 B_1}$, we conclude that
	\begin{align}
		\frac12 \left( 1 - \left\| p \rho_{R_1 B_1} - q \tau_{R_1B_1} \right\|_1 \right)
		&\leq p^s q^{1-s} \Tr\!\left[ \rho_{R_1 B_1}^s \tau_{R_1 B_1}^{1-s} \right]
		\\
		&\leq p^s q^{1-s}
		\\
		&\leq \varepsilon.
	\end{align}
	In the above, we used the H\"{o}lder inequality $\left|\Tr[XY] \right| \leq \left\|X\right\|_r \left\|Y\right\|_t$, which holds for $r,t\geq 1$, such that $r^{-1} + t^{-1} = 1$, to conclude that
 $\Tr\!\left[ \rho_{R_1 B_1}^s \tau_{R_1 B_1}^{1-s} \right]\leq 1.$
	That is, set $r=s^{-1}, t = (1-s)^{-1}$, $X=\rho_{R_1B_1}$, and $Y=\tau_{R_1 B_1}$.
	The last inequality follows by the assumption.
	So we conclude that only a single access to the channel is needed in this case also.

 \underline{Third condition:} Assume that $\exists$ $\psi_{RA}$ such that
        $\operatorname{supp}\!\left( \cN_{A \to B} (\psi_{RA}) \right) \cap \operatorname{supp}\!\left( \cM_{A \to B} (\psi_{RA}) \right) = \emptyset$. 
    In this case, choose the initial state to be $\rho_{R_1 A_1} =\psi_{RA}$. Then the minimum error evaluates to be $\frac12 \left( 1 - \left\| p \cN_{A \to B} (\psi_{RA}) - q \cM_{A \to B} (\psi_{RA}) \right\|_1 \right)$. Since the supports of the two output states are disjoint, we conclude that $\left\| p \cN_{A \to B} (\psi_{RA}) - q \cM_{A \to B} (\psi_{RA}) \right\|_1  =p+q =1$. Then, we achieve zero error by following this procedure and the choice of the initial state with just one query to the unknown channel, concluding the proof.

 \medskip
    We finally prove~\eqref{eq:binary_symmetric2_C}.
	This follows because it is impossible to distinguish the states and the desired inequality $\frac12 \left( 1 - \left\| p \rho_{R_n B_n} - q \tau_{R_n B_n} \right\|_1 \right) \leq \varepsilon$ cannot be satisfied for any possible value of~$n$.
	Indeed, in this case, since $\cN= \cM$ we have 
 $\rho_{R_n B_n} =\tau_{R_n B_n}$ and 
	\begin{align}
		\frac12 \left( 1 - \left\| p \rho_{R_n B_n} - q \tau_{R_n B_n} \right\|_1 \right)
		&= \frac12 \left(1 - |p-q| \left\|\rho_{R_n B_n}\right\|_1 \right)
		\\
		&= \frac12 \left( 1 - |p-(1-p)| \right)
		\\
		&= \frac12 ( 1 - |1-2p|) \\
  &=\min\{p,q\}.
	\end{align}
Then, with the assumption $\min\{p,q\} \geq \varepsilon$, for all $n \in \cN$, the error probability is always greater than or equal to $\varepsilon$. Thus,~\eqref{eq:binary_symmetric2_C} follows.

\section{Proof of Theorem~\ref{theorem:binary_symmetric_C}}
\label{app:binary_symmetric_C}

\underline{Upper bound:} We first set
\begin{equation}
n\coloneqq \left\lceil \inf_{s\in\left[  0,1\right]  }\frac{\ln\!\left(  \frac
{p^{s}q^{1-s}}{\varepsilon}\right)  }{C_s(\cN \Vert \cM)}\right\rceil ,\label{eq:choice-samp-comp-upp-bnd}
\end{equation}
and let $s^{\ast}\in\left[  0,1\right]  $ be the optimal value of $s$ above.

With $C_s(\cN \Vert \cM) \geq 0$, due to $-\ln \Tr\!\left[ \rho^s \sigma^{1-s} \right]\geq 0$ for all states $\rho$ and $\sigma$ and $s \in [0,1]$,
then
\begin{align}
n  & \geq\frac{\ln\!\left(  \frac{p^{s^{\ast}}q^{1-s^{\ast}}}{\varepsilon
}\right)  }{C_{s^{\ast}}(\cN \Vert \cM)}\\
\Leftrightarrow\quad nC_{s^{\ast}}(\cN \Vert \cM)  &
\geq\ln\!\left(  p^{s^{\ast}}q^{1-s^{\ast}}\right)  -\ln\varepsilon
\end{align}
\begin{align}
\Leftrightarrow\quad- \inf_{\rho_{RA}}\ln\operatorname{Tr}\!\left[\left( \left(\cN( \rho_{RA})\right)^{\otimes n}\right)
^{s^{\ast}}\left(  \left(\cM(\rho_{RA})\right)^{\otimes n}\right)  ^{1-s^{\ast}}\right]-\ln\!\left(
p^{s^{\ast}}q^{1-s^{\ast}}\right)    & \geq-\ln\varepsilon\\
\Leftrightarrow\quad\inf_{\rho_{RA}}{\Tr}\!\left[\left(  p\left(\cN(\rho_{RA})\right)^{\otimes n}\right)
^{s^{\ast}}\left(  q\left(\cM(\rho_{RA})\right)^{\otimes n}\right)  ^{1-s^{\ast}}\right]  &
\leq\varepsilon.
\end{align}
Now applying the quantum Chernoff bound in~\cite[Theorem~1]{ACM+07} 
with $A = p \rho^{\otimes n}$ and $B = q\sigma^{\otimes n}$, we conclude that
\begin{align}
\varepsilon &\geq  \inf_{\rho_{RA}} \operatorname{Tr}\!\left[\left(  p \cN^{\otimes n} (\rho_{RA}^{\otimes n})^{\otimes n}\right)
^{s^{\ast}}\left(  q\cM^{\otimes n} (\rho_{RA}^{\otimes n})^{\otimes n}\right)  ^{1-s^{\ast}}\right]\\
& \geq\inf_{\rho_{RA}} \frac12 \left( 1 - \left\| p \cN^{\otimes n} (\rho_{RA}^{\otimes n}) - q \cM^{\otimes n} (\rho_{RA}^{\otimes n}) \right\|_1 \right) .
\label{eq:binary_symmetric_achievability_final}
\end{align}

Consider that 
 \begin{equation}
      p_e(p, \cN,q,\cM,n) = \inf_{\{ Q,\cA\},\rho_{R_1A_1}} \frac12 \left( 1 - \left\| p \rho_{R_n B_n} - q \tau_{R_n B_n} \right\|_1 \right).
 \end{equation}
 Choosing $n$ copies of the input state and passing it through channel $\cN$ or $\cM$ in its $n$ uses, the final state before the measurement is applied evaluates to be either of these states ($\cN^{\otimes n} (\rho_{RA}^{\otimes n}) $ or $ \cM^{\otimes n} (\rho_{RA}^{\otimes n}) $). This is one strategy that is included in the set of all adaptive strategies. This leads to the following upper bound:
 \begin{align} \label{eq:one}
     p_e(p, \cN,q,\cM,n) & \leq \inf_{\rho_{RA}} \frac12 \left( 1 - \left\| p \cN^{\otimes n} (\rho_{RA}^{\otimes n}) - q \cM^{\otimes n} (\rho_{RA}^{\otimes n}) \right\|_1 \right).
 \end{align}

The choice of $n$ in~\eqref{eq:choice-samp-comp-upp-bnd} thus satisfies the
constraint~\eqref{eq:def:sample_complexity_symmetric2_C} in Definition~\ref{def:binary_symmetric_C}, and since the optimal query complexity
cannot exceed this choice, this concludes our proof of the upper bound in~\eqref{eq:binary_symmetric3_C}.

\bigskip 
\underline{Lower bound 1:}
Let us consider the case where 
\begin{equation}
    \varepsilon \geq  p_e(p, \cN,q,\cM,n) = \inf_{\{ \cA\},\rho_{R_1A_1}} \frac12 \left( 1 - \left\| p \rho_{R_n B_n} - q \tau_{R_n B_n} \right\|_1 \right). \label{eq:start}
\end{equation}
Applying \eqref{eq:take-off-step-channels}, we conclude that
\begin{align}
    \frac{\varepsilon\left(1-\varepsilon\right)}{pq}
    &\geq \inf_{\{ \cA\},\rho_{R_1A_1}}  F\!\left( \rho_{R_n B_n}, \tau_{R_n B_n}\right) \label{eq:start_for other proofs} \\
    &\geq \inf_{\{ \cA\},\rho_{R_1A_1}}  \widehat{F}\!\left( \rho_{R_n B_n}, \tau_{R_n B_n}\right), \label{eq:combining}
\end{align}
where the last inequality follows because $\widehat{F} \leq F$.

Recall that, for $1\leq i \leq n$,
\begin{align}
    \rho_{R_{i} B_{i}} & \coloneqq \cN_{A_i \to B_i} (\rho_{R_i A_i}), \label{eq:similar_start}\\ 
    \tau_{R_{i} B_{i}} & \coloneqq  \cM_{A_i \to B_i} (\tau_{R_i A_i}),
\end{align}
and 
\begin{align}
    \rho_{R_{i+1} A_{i+1}} & \coloneqq (\cA_{R_i B_i \to R_{i+1} A_{i+1}}^{(i)} \circ \cN_{A_i \to B_i}) (\rho_{R_i A_i}), \\ 
    \tau_{R_{i+1} A_{i+1}} & \coloneqq (\cA_{R_i B_i \to R_{i+1} A_{i+1}}^{(i)} \circ \cM_{A_i \to B_i}) (\tau_{R_i A_i}),
\end{align}
for every $1 \leq i < n$.
With the above notation, and by using the chain rule for geometric fidelity in~\eqref{eq:chain_rule_fide}, we have that 
\begin{align}
   \widehat{F}\!\left( \rho_{R_n B_n}, \tau_{R_n B_n}\right) &=   \widehat{F}\!\left( \cN_{A_n \to B_n} (\rho_{R_n A_n}), \cM_{A_n \to B_n} (\rho_{R_n A_n})\right)    \\ 
   &\geq \widehat{F}(\cN, \cM) \widehat{F}\!\left( \rho_{R_n A_n}, \tau_{R_n A_n}\right).
\end{align}
Then using the data-processing inequality for the geometric fidelity under the channel $\cA_{R_{n-1} B_{n-1} \to R_{n} A_{n}}^{(n-1)}$, we have that 
\begin{align}
    \widehat{F}\!\left( \rho_{R_n A_n}, \tau_{R_n A_n}\right) & \geq \widehat{F}\!\left( \cN_{A_{n-1} \to B_{n-1}}(\rho_{R_{n-1} A_{n-1}}), \cM_{A_{n-1} \to B_{n-1}}(\tau_{R_{n-1} A_{n-1}})\right).
\end{align}
Proceeding with the application of the chain rule and data processing under the adaptive channels for $n-1$ times, we obtain 
\begin{align}
    \widehat{F}\!\left( \rho_{R_n B_n}, \tau_{R_n B_n}\right) &\geq \left(\widehat{F}(\cN, \cM)\right)^{n-1} \widehat{F}\!\left( \rho_{R_1 B_1}, \tau_{R_1 B_1}\right) \\
    &= \left(\widehat{F}(\cN, \cM)\right)^{n-1} \widehat{F}\!\left( \cN_{A_1 \to B_1}(\rho_{R_1 A_1}), \cM_{A_1 \to B_1}(\tau_{R_1 A_1})\right). \label{eq:F_hat_channel_bound}
\end{align}

Combining the above inequality with~\eqref{eq:combining}, we arrive at 
\begin{align}
     \frac{\varepsilon\left(1-\varepsilon\right)}{pq}  
    & \geq    \left(\widehat{F}(\cN, \cM)\right)^{n-1} \inf_{\{ \cA\},\rho_{R_1A_1}} \widehat{F}\!\left( \cN_{A_1 \to B_1}(\rho_{R_1 A_1}), \cM_{A_1 \to B_1}(\tau_{R_1 A_1})\right) \\ & =  \left[\widehat{F}(\cN, \cM)\right]^{n}, \label{eq:bound_on-error}
\end{align}
where the last equality follows by the definition of channel fidelity with the choice $\widehat{F}$ in~\eqref{eq:channel_fidelities}.

After taking a logarithm and rearranging, we conclude the proof of the first lower bound in~\eqref{eq:binary_symmetric3_C}.

\bigskip 
\underline{Lower bound 2:} 
Using \cite[Eq.~(II.22)]{cheng2024invitation}, we have that 
\begin{align}
\left\Vert p\rho-q\sigma\right\Vert _{1}^{2}  &  \leq1-4pqF(\rho,\sigma)\\
&  \leq1-4pq+4pq\left[  d_B(\rho,\sigma)\right]  ^{2} \\
& \leq1-4pq+4pq\left[  d_{\widehat{F}}(\rho,\sigma)\right]  ^{2} \\
&= 1-4pq+8pq \left(1-\left(1-  \frac{\left[  d_{\widehat{F}}(\rho,\sigma)\right]  ^{2}}{2}\right) \right),
\end{align}
where the third inequality follows because $\widehat{F} \leq F$.

Rewriting~\eqref{eq:start} by algebraic manipulations, we get 
\begin{align}
    (1-2\varepsilon)^2 &\leq \sup_{\{ \cA\},\rho_{R_1A_1}}  \left\| p \rho_{R_n B_n} - q \tau_{R_n B_n} \right\|_1^2 \\
    & \leq 1- 4pq +8pq \left(1- \inf_{\{ \cA\},\rho_{R_1A_1}}  \left(1-  \frac{\left[  d_{\widehat{F}}(\rho_{R_n B_n},\tau_{R_n B_n})\right]  ^{2}}{2}\right) \right) \\ 
    &\leq  1- 4pq +8pq \left(1-   \left(1-  \frac{\left[  d_{\widehat{F}}(\cN , \cM)\right]  ^{2}}{2}\right)^{n} \right) \\
    & \leq 1-4pq+4pqn \left[  d_{\widehat{F}}(\cN, \cM)\right]  ^{2}, \label{eq:last_before_SC}
\end{align}
where the last inequality follows because $1- (1-x)^n \leq nx$ for $0\leq x \leq 1$. Now, we prove the penultimate inequality: first observe that 
\begin{equation}
    1- \frac{\left[  d_{\widehat{F}}(\rho,\sigma)\right]  ^{2}}{2} = \widehat{F}(\rho,\sigma) .
\end{equation}
This relation is also true for the channel variants  
\begin{equation}\label{eq:channel_d_b_hat}
    1- \frac{\left[  d_{\widehat{F}}(\cN,\cM)\right]  ^{2}}{2} = \widehat{F}(\cN,\cM) .
\end{equation}
Using~\eqref{eq:F_hat_channel_bound}, we have that 
\begin{align}
   \widehat{F}\!\left( \rho_{R_n B_n}, \tau_{R_n B_n}\right) &\geq \left[\widehat{F}(\cN, \cM)\right]^{n-1} \widehat{F}\!\left( \cN_{A_1 \to B_1}(\rho_{R_1 A_1}), \cM_{A_1 \to B_1}(\tau_{R_1 A_1})\right) \\
   & \geq \left[\widehat{F}(\cN, \cM)\right]^{n} \\ 
   &= \left( 1- \frac{\left[  d_{\widehat{F}}(\cN,\cM)\right]  ^{2}}{2}  \right)^n,
\end{align}
where the last inequality follows by the definition of channel fidelities in~\eqref{eq:channel_fidelities}, and the last equality by~\eqref{eq:channel_d_b_hat}. With that, we establish that 
\begin{equation}
   \inf_{\{ \cA\},\rho_{R_1A_1}}  \left(1-  \frac{\left[  d_{\widehat{F}}(\rho_{R_n B_n},\tau_{R_n B_n})\right]  ^{2}}{2}\right) \geq  \left( 1- \frac{\left[  d_{\widehat{F}}(\cN,\cM)\right]  ^{2}}{2}  \right)^n.
\end{equation}
This completes the proof of~\eqref{eq:last_before_SC}.

By rearranging terms in \eqref{eq:last_before_SC}, we obtain the following:
\begin{align}
    n & \geq \frac{(1-2\varepsilon)^2 -(1-4pq)}{4pq\left[  d_{\widehat{F}}(\cN,\cM)\right]  ^{2}} \\
    &\geq \frac{pq -\varepsilon(1-\varepsilon)}{pq\left[  d_{\widehat{F}}(\cN,\cM)\right]  ^{2}},\\
    &= \frac{1 -\frac{\varepsilon(1-\varepsilon)}{pq}}{\left[  d_{\widehat{F}}(\cN,\cM)\right]  ^{2}},
\end{align}
which concludes the proof of the second lower bound in~\eqref{eq:binary_symmetric3_C}.

\section{Proof of Theorem~\ref{thm:binary_asymmetric_C}}

\label{app:proof_samp_comp_binary_asymmetric_C}

\underline{Lower bound:}

Let $\alpha \in (1,2]$. Then, by the data-processing inequality for the geometric R\'enyi divergence under the measurement channel (comprised of the POVM elements $Q, I-Q$)
\begin{align}
    \widehat{D}_\alpha (\rho\Vert \tau) &\geq \frac{1}{\alpha-1} \ln\! \left( \Tr[Q \rho]^\alpha \Tr[Q\tau]^{1-\alpha} + \Tr[(I-Q)\rho]^\alpha + \Tr[(I-Q) \tau]^{1-\alpha} \right) \\
    & \geq \frac{\alpha}{\alpha-1} \ln \Tr[Q\rho] - \ln \Tr[Q \tau].
\end{align}

Let $\rho =\rho_{R_n B_n}$, $\tau=\tau_{R_n B_n}$ with the constraint that $\Tr[Q_{R_n B_n} \rho_{R_n B_n}] \geq 1- \varepsilon$. Then, the above inequality can be rewritten as 
\begin{equation}
     \widehat{D}_\alpha (\rho_{R_n B_n}\Vert \tau_{R_n B_n}) \geq \frac{\alpha}{\alpha-1} \ln(1-\varepsilon) -  \ln \Tr[Q_{R_n B_n} \tau_{R_n B_n}].
\end{equation}

From here, let us optimize over all $\{ Q, \cA\}$ and $\rho_{R_1 A_1}$ such that 
\begin{align}
   & \sup_{\{ Q, \cA\}, \rho_{R_1 A_1}} \widehat{D}_\alpha (\rho_{R_n B_n}\Vert \tau_{R_n B_n})  \notag \\
    &\geq \frac{\alpha}{\alpha-1} \ln(1-\varepsilon) -  \ln  \inf_{\{ Q, \cA\}, \rho_{R_1 A_1}} \left\{ \Tr[Q_{R_n B_n} \tau_{R_n B_n}] : \Tr[Q_{R_n B_n} \rho_{R_n B_n}] \geq 1- \varepsilon \right\} \\ 
    & = \frac{\alpha}{\alpha-1} \ln(1-\varepsilon) -  \ln \beta_\varepsilon\!\left( \cN^{(n)} \Vert \cM^{(n)}\right).
\end{align}

By assuming $\beta_\varepsilon\!\left( \cN^{(n)} \Vert \cM^{(n)}\right) \leq \delta$, we also have that 
\begin{align}
     \sup_{\{ Q, \cA\}, \rho_{R_1 A_1}} \widehat{D}_\alpha (\rho_{R_n B_n}\Vert \tau_{R_n B_n})  &\geq \frac{\alpha}{\alpha-1} \ln(1-\varepsilon) -  \ln \delta \\ 
     &= \ln\!\left( \frac{(1-\varepsilon)^{\alpha'}}{\delta} \right), \label{eq:bound_delta_eps}
\end{align}
where $\alpha'$ is defined in the  statement of Theorem~\ref{thm:binary_asymmetric_C}.

By using the chain rule for the geometric R\'enyi divergence in~\eqref{eq:chain_rule_geometric_renyi}, we have that 
\begin{align}
   \widehat{D}_\alpha\!\left( \rho_{R_n B_n} \Vert \tau_{R_n B_n}\right) &=   \widehat{D}_\alpha\!\left( \cN_{A_n \to B_n} (\rho_{R_n A_n}) \Vert \cM_{A_n \to B_n} (\rho_{R_n A_n})\right)    \\ 
   &\leq \widehat{D}_\alpha(\cN \Vert \cM) + \widehat{D}_\alpha\!\left( \rho_{R_n A_n} \Vert \tau_{R_n A_n}\right).
\end{align}
Then using the data-processing inequality for the geometric R\'enyi divergence under the channel  $\cA_{R_{n-1} B_{n-1} \to R_{n} A_{n}}^{(n-1)}$, we have that 
\begin{align}
    \widehat{D}_\alpha\!\left( \rho_{R_n A_n} \Vert \tau_{R_n A_n}\right) & \leq \widehat{D}_\alpha\!\left( \cN_{A_{n-1} \to B_{n-1}}(\rho_{R_{n-1} A_{n-1}}) \Vert\cM_{A_{n-1} \to B_{n-1}}(\tau_{R_{n-1} A_{n-1}})\right).
\end{align}
Proceeding with the application of the chain rule and the data-processing inequality under the adaptive channels for $n-1$ times, we obtain 
\begin{align}
    \widehat{D}_\alpha\!\left( \rho_{R_n B_n} \Vert \tau_{R_n B_n}\right) &\leq (n-1)\widehat{D}_\alpha(\cN \Vert \cM) + \widehat{D}_\alpha\!\left( \rho_{R_1 B_1} \Vert \tau_{R_1 B_1}\right) \\
    &= (n-1)\widehat{D}_\alpha(\cN \Vert \cM) + \widehat{D}_\alpha\!\left( \cN_{A_1 \to B_1}(\rho_{R_1 A_1})\Vert \cM_{A_1 \to B_1}(\tau_{R_1 A_1})\right) \\
    &\leq n \widehat{D}_\alpha(\cN \Vert \cM) \label{eq:bound_D_hat_output},
\end{align}
where the last inequality follows from the definition of the channel variant of the geometric R\'enyi divergence obtained by replacing $\bm{D}$ in~\eqref{eq:channel_divergences} with $\widehat{D}_\alpha$.

Combining~\eqref{eq:bound_delta_eps} and~\eqref{eq:bound_D_hat_output}, we arrive at the following for all $\alpha \in (1,2)$:
\begin{equation}
  n \widehat{D}_\alpha(\cN \Vert \cM)  \geq   \ln\!\left( \frac{(1-\varepsilon)^{\alpha'}}{\delta} \right),
\end{equation}
which leads to 
\begin{equation}
    n \geq \sup_{\alpha \in (1,2)} \frac{\ln\!\left( \frac{(1-\varepsilon)^{\alpha'}}{\delta} \right)}{\widehat{D}_\alpha(\cN \Vert \cM) }.
\end{equation}

The other lower bound is obtained by switching the roles of $\cN$ and $\cM$, and $\varepsilon$ with $\delta$ with the use of the fact 
\begin{equation}
n^{\ast}(\cN,\cM,\varepsilon,\delta)\coloneqq \inf\left\{  n\in\mathbb{N}
:\beta_{\delta} \!\left(\cM^{(n)}\Vert \cN^{(n)} \right)\leq
\varepsilon\right\}  ,
\end{equation}
as given in~\eqref{eq:asymm-beta-rewrite-2}.

\bigskip 
\underline{Upper bound:}
By~\cite[Lemma~A.3 of the supplemental material]{cheng2024invitation}, for all $\alpha \in (0,1)$ and $\varepsilon \in (0,1)$, we have that
\begin{equation}
    - \ln \inf_{Q}\{ \Tr[Q \sigma]: \Tr[Q\rho] \geq 1-\varepsilon \} \geq D_\alpha(\rho \Vert \sigma) + \frac{\alpha}{\alpha-1} \ln\!\left( \frac{1}{\varepsilon}\right).
\end{equation}

With that for our setup, let us consider $\rho=\rho_{R_n B_n}$ and $\sigma= \tau_{R_n B_n}$ and $Q=Q_{R_n B_n}$ and supremize over all adaptive strategies $\{Q,\cA\}$ and input states $\rho_{R_1A_1}$ to obtain 
\begin{align}
-\ln \beta_\varepsilon\!\left( \cN^{(n)} \Vert \cM^{(n)}\right)   &=-  \ln  \inf_{\{ Q, \cA\}, \rho_{R_1 A_1}} \left\{ \Tr[Q_{R_n B_n} \tau_{R_n B_n}] : \Tr[Q_{R_n B_n} \rho_{R_n B_n}] \geq 1- \varepsilon \right\}   \\
  &\geq  \sup_{\{ Q, \cA\}, \rho_{R_1 A_1}} {D}_\alpha (\rho_{R_n B_n}\Vert \tau_{R_n B_n}) + \frac{\alpha}{\alpha-1} \ln\!\left( \frac{1}{\varepsilon}\right). \label{eq:bound_cN_cM}
\end{align}

Consider a product strategy, where we discriminate between the states $[\cN(\rho_{RA})]^{\otimes n}$ and $[\cM(\rho_{RA})]^{\otimes n}$ with the input state being $\rho_{RA}$ and optimizing over such strategies with different input states is included in the set of all possible strategies. This provides the following lower bound:
\begin{align}
    \sup_{\{ Q, \cA\}, \rho_{R_1 A_1}} {D}_\alpha (\rho_{R_n B_n}\Vert \tau_{R_n B_n}) \geq \sup_{\{ Q, \cA\}, \rho_{R A}} {D}_\alpha \!\left( [\cN(\rho_{RA})]^{\otimes n} \Vert [\cM(\rho_{RA})]^{\otimes n} \right) \\
    = n \cdot \sup_{\{ Q, \cA\}, \rho_{R A}} {D}_\alpha \!\left( \cN(\rho_{RA}) \Vert \cM(\rho_{RA}) \right),\label{eq:bound_additive_n}
\end{align}
where the last equality follows by the additivity of the Petz-R\'enyi relative entropy.

Combining~\eqref{eq:bound_cN_cM} and~\eqref{eq:bound_additive_n}, we have that 
\begin{equation}
  -\ln \beta_\varepsilon\!\left( \cN^{(n)} \Vert \cM^{(n)}\right) \geq   n \cdot \sup_{\{ Q, \cA\}, \rho_{R A}} {D}_\alpha \!\left( \cN(\rho_{RA}) \Vert \cM(\rho_{RA}) \right) + \frac{\alpha}{\alpha-1} \ln\!\left( \frac{1}{\varepsilon}\right),
\end{equation}
for all $\alpha \in (0,1)$.

Using that, with the choice
\begin{equation}
    n \geq \inf_{\alpha\in\left(  0,1\right)  }\left(  \frac{\ln\!\left(  \frac
{\varepsilon^{\alpha^{\prime}}}{\delta}\right)  }{D_{\alpha}(\cN\Vert\cM
)}\right),
\end{equation}
we have that $\beta_\varepsilon\!\left( \cN^{(n)} \Vert \cM^{(n)}\right) \leq \delta$. 
Then, the optimum number of channel uses with this strategy is
\begin{equation}
    n= \left\lceil \inf_{\alpha\in\left(  0,1\right)  }\left(  \frac{\ln\!\left(  \frac
{\varepsilon^{\alpha^{\prime}}}{\delta}\right)  }{D_{\alpha}(\cN\Vert\cM
)}\right) \right\rceil,
\end{equation}
which leads to 
\begin{equation}
  n^{\ast}(\cN,\cM,\varepsilon,\delta)   \leq \left\lceil \inf_{\alpha\in\left(  0,1\right)  }\left(  \frac{\ln\!\left(  \frac
{\varepsilon^{\alpha^{\prime}}}{\delta}\right)  }{D_{\alpha}(\cN\Vert\cM
)}\right) \right\rceil.
\end{equation}

 \section{Proof of Theorem~\ref{theorem:sc_M-ary_C}}

\label{app:proof-sc_M-ary_C}

\underline{Lower bound:}
Recall~\eqref{eq:error_M-ary_C}, so that
\begin{align}
p_e(\mathcal{S},n) 
	&\coloneqq \inf_{\{ Q, \cA\},\rho_{R_1A_1}} \sum_{m=1}^M p_m \Tr\!\left[ (I_{R_n B_n} -Q_{R_n B_n}^m )\rho_{R_n B_n}^m\right],
\end{align}
where $Q$ denotes a POVM $(Q^1_{R_n B_n},\ldots, Q^M_{R_n B_n})$ satisfying $Q^m_{R_n B_n} \geq 0$ for all $m\in \{1,\ldots,M\}$ and $\sum_{m=1}^M {Q}^m_{R_n B_n} = I_{R_n B_n}$.

    For $1\leq m \neq  \tilde{m}\leq M$, choose $L_{R_n B_n}$ and $T_{R_n R_n}$ to be positive semi-definite operators satisfying $L_{R_n B_n} + T_{R_n B_n}= I_{R_n B_n}- Q^m_{R_n B_n}- Q^{\tilde{m}}_{R_nB_n}$. With that, define $\tilde{Q}_{R_n B_n}^{m} \coloneqq  Q^m_{R_n B_n}+ L_{R_n B_n}$ and $\tilde{Q}^{\tilde{m}}_{R_n B_n} \coloneqq  Q^{\tilde{m}}_{R_n B_n}+ T_{R_n B_n}$, so that we have $\tilde{Q}^{{m}}_{R_n B_n}+\tilde{Q}^{\tilde{m}}_{R_n B_n}=I$. 

Consider
\begin{align}
    & \sum_{m=1}^M p_m \Tr\!\left[ (I_{R_n B_n} -Q_{R_n B_n}^m )\rho_{R_n B_n}^m\right] \notag \\
   & \geq  p_m \Tr\!\left[ (I_{R_n B_n} -Q_{R_n B_n}^m )\rho_{R_n B_n}^m\right] + p_{\tilde{m}} \Tr\!\left[ (I_{R_n B_n} -Q_{R_n B_n}^{\tilde{m}} )\rho_{R_n B_n}^{\tilde{m}}\right] \\
   & \geq p_m\Tr\!\left[ (I_{R_n B_n} -\tilde{Q}_{R_n B_n}^m )\rho_{R_n B_n}^m\right] + p_{\tilde{m}} \Tr\!\left[ (I_{R_n B_n} -\tilde{Q}_{R_n B_n}^{\tilde{m}} )\rho_{R_n B_n}^{\tilde{m}}\right] \\
   &= (p_m +p_{\tilde{m}}) \left(\frac{p_m}{p_m +p_{\tilde{m}}}\Tr\!\left[ (I_{R_n B_n} -\tilde{Q}_{R_n B_n}^m )\rho_{R_n B_n}^m\right] + \frac{p_{\tilde{m}}}{p_m +p_{\tilde{m}}} \Tr\!\left[ (I_{R_n B_n} -\tilde{Q}_{R_n B_n}^{\tilde{m}} )\rho_{R_n B_n}^{\tilde{m}}\right] \right) \\
   &\geq (p_m +p_{\tilde{m}}) \inf_{\{\tilde{Q},\cA\}, \rho_{R_1A_1}} \left(\frac{p_m}{p_m +p_{\tilde{m}}}\Tr\!\left[ \tilde{Q}_{R_n B_n}\rho_{R_n B_n}^m\right] + \frac{p_{\tilde{m}}}{p_m +p_{\tilde{m}}} \Tr\!\left[ (I_{R_n B_n} -\tilde{Q}_{R_n B_n} )\rho_{R_n B_n}^{\tilde{m}}\right] \right) \\
   & \geq \frac{p_m p_{\tilde{m}}}{p_m+p_{\tilde{m}}} \left(\widehat{F}(\cN^m,\cN^{\tilde{m}} \right)^n,
\end{align}
where the second inequality follows by the explicit construction of $\tilde{Q}$ so that 
 ${Q}_{R_n B_n}^i  \leq \tilde{Q}_{R_n B_n}^i $ for $i \in \{m,\tilde{m}\}$ and the last inequality by applying~\eqref{eq:bound_on-error}.

By optimizing the left-hand side over all $\{Q,\cA\}$ and $\rho_{R_1A_1}$ and imposing the constraint on the error, we have that 
\begin{align}
    \varepsilon &\geq p_e(\mathcal{S},n)  \\
    & \geq \frac{p_m p_{\tilde{m}}}{p_m+p_{\tilde{m}}} \left(\widehat{F}(\cN^m,\cN^{\tilde{m}}) \right)^n.
\end{align}
Rearranging the terms in the above inequality and maximizing over all arbitary pairs $m,\tilde{m}$ such that $1\leq m \neq\tilde{m} \leq M$ concludes the proof of the desired lower bound.

\bigskip
\underline{Upper bound:}
By employing the product strategy of discriminating the following states $\{ (\cN^m_{A \to B}(\rho_{RA}))^{\otimes n}\}_m$ with the initial state $\rho_{RA}$, we have that

\begin{align}
    p_e(\mathcal{S},n)  & \leq \inf_{Q, \rho_{RA}} \sum_{m=1}^M p_m \Tr\!\left[ (I_{R_n B_n} -Q_{R_n B_n}^m ) \left( \cN^m(\rho_{RA})\right)^{\otimes n}\right]  \\
    & \leq \inf_{Q, \rho_{RA}} 
    \frac{1}{2} M(M-1) \max_{\bar{m}\neq m}  \sqrt{p_m}\sqrt{p_{\bar{m}}} \sqrt{F}\left( \cN^m(\rho_{RA}), \cN^{\bar{m}}(\rho_{RA})  \right) \\ 
    & = \frac{1}{2} M(M-1) \max_{\bar{m}\neq m}  \sqrt{p_m}\sqrt{p_{\bar{m}}} \sqrt{F}\left( \cN^m, \cN^{\bar{m}} \right), \label{eq:error_bound_product_M}
\end{align}
where the second inequality follows by applying \cite[Eq.~(IV.145)]{cheng2024invitation} to this setting and the last equality by the definition of channel fidelities in~\eqref{eq:channel_fidelities} with the choice $F$.

Then, by choosing 
\begin{equation}
    n \geq \max_{m\neq \bar{m}}  \frac{2\ln\!\left( \frac{ M(M-1) \sqrt{p_m} \sqrt{p_{\bar{m}}}  }{ 2\varepsilon } \right) }{ - \ln F\!\left(\cN^{m},\cN^{\bar{m}} \right) },
\end{equation}
we obtain that $ p_e(\mathcal{S},n) \leq \varepsilon$ using~\eqref{eq:error_bound_product_M}. 
The optimum $n$ to choose with this strategy would be 
\begin{equation}
    n= \left \lceil \max_{m\neq \bar{m}}  \frac{2\ln\!\left( \frac{ M(M-1) \sqrt{p_m} \sqrt{p_{\bar{m}}}  }{ 2\varepsilon } \right) }{ - \ln F\!\left(\cN^{m},\cN^{\bar{m}} \right) }  \right \rceil. 
\end{equation}
With this, we conclude that 
\begin{equation}
    n^*(\mathcal{S},\varepsilon) \leq 
		\left\lceil  \max_{m\neq \bar{m}}  \frac{2\ln\!\left( \frac{ M(M-1) \sqrt{p_m} \sqrt{p_{\bar{m}}}  }{ 2\varepsilon } \right) }{ - \ln F\!\left(\cN^{m},\cN^{\bar{m}} \right) } \right\rceil .
\end{equation}

\end{document}